\documentclass[lettersize,journal]{IEEEtran}
\usepackage{amsmath,amssymb,amsfonts}
\usepackage{amsthm}
\usepackage{algorithmic}
\usepackage{algorithm}
\usepackage{array}
\usepackage[caption=false,font=footnotesize,labelfont=rm,textfont=rm]{subfig}

\usepackage{textcomp}
\usepackage{stfloats}
\usepackage{url}
\usepackage{verbatim}
\usepackage{graphicx}
\usepackage{cite}
\usepackage{graphicx}
\usepackage{textcomp}
\usepackage{xcolor}
\usepackage{multirow}
\usepackage{makecell}
\usepackage{booktabs}
\newtheorem{thm}{Theorem}
\newcounter{TempEqCnt}

\begin{document}

\title{Coherent FDA Radar: Transmitter and Receiver Design and Analysis}

\author{Yan Sun, Ming-jie Jia, Wen-qin Wang,~\IEEEmembership{Senior member,~IEEE}, Maria Sabrina Greco,~\IEEEmembership{Fellow,~IEEE}, Fulvio Gini,~\IEEEmembership{Fellow,~IEEE} and Shunsheng Zhang
\thanks{The work of Y. Sun, M.-J Jia, W.-Q. Wang, and S. Zhang was supported by the National Natural Science Foundation of China under Grant 62171092. The work of F. Gini and M. S. Greco has been partially supported by the Italian Ministry of Education and Research (MUR) in the framework of the FoReLab project (Departments of Excellence).~\itshape (Corresponding author: Wen-qin Wang)}
\thanks{Yan Sun, Ming-jie Jia, Wen-qin Wang and Shunsheng Zhang are with University of Electronic Science and Technology of China, Chengdu 611731, China (e-mail: sunyan\_1995@163.com; jiamingjie@std.uestc.edu.cn; wqwang@uestc.edu.cn; zhangss@uestc.edu.cn).}
\thanks{Maria Sabrina Greco and Fulvio Gini are with University of Pisa, Pisa, Italy (e-mail: maria.greco@unipi.it, fulvio.gini@unipi.it).}}



\maketitle

\begin{abstract}
The combination of frequency diverse array (FDA) radar technology with the multiple input multiple output (MIMO) radar architecture and waveform diversity techniques potentially promises a high integration gain with respect to conventional phased array (PA) radars. In this paper, we propose an approach to the design of the transmitter and the receiver of a coherent FDA (C-FDA) radar, that enables it to perform the demodulation with spectral overlapping, due to the small frequency offset. To this purpose, we derive the generalized space-time-range signal model and we prove that the proposed C-FDA radar has a higher coherent array gain than a PA radar, and at the same time, it effectively resolves the secondary range-ambiguous (SRA) problem of FDA-MIMO radar, allowing for mainlobe interference suppression and range-ambiguous clutter suppression. Numerical analysis results prove the effectiveness of the proposed C-FDA radar in terms on anti-interference and anti-clutter capabilities over conventional radars.
 
\end{abstract}


\section{Introduction}

\IEEEPARstart{M}{ultiple}-antennas and multiple-channels radar systems have been widely employed in challenging detection missions with significant developments in array signal processing. As conventional radars, phased-array (PA) radar \cite{ref1} and multiple-input multiple-output (MIMO) radar \cite{ref2}\cite{ref3} have been extensively investigated due to their advantages, especially the high coherent array gain for PA radar \cite{ref4} and the high spatial resolution brought by the virtual array for MIMO radar \cite{ref5}. Moreover, inspired by array diversity, frequency diverse array (FDA) radar \cite{ref6} has been proposed to break through the application limitations of conventional radars due to its range-dependent beam \cite{ref7}. After almost two decades of research, FDA-MIMO radar \cite{ref8}, requiring a frequency offset not less than the bandwidth to ensure waveform orthogonality, has been extensively investigated for airborne radar\cite{ref9}, vehicular radar\cite{ref10}, sonar\cite{ref11}, and integrated sensing and communication system\cite{ref12}\cite{ref13}. Apart from the benefit of range-dependency, FDA-MIMO radar enjoys a higher spatial resolution and more array structural degrees of freedom (DOFs) inherited from MIMO radar. However, such merits are achieved at the cost of losing the transmit signal coherency, resulting in a lower array gain and less robustness to noisy background than PA radar.

Aiming at increasing the coherent processing gain to yield a larger signal-to-noise (SNR) for target detection, many approaches to the design of FDA-MIMO radar have been developed in terms of subarray design, space-time coding, and transceiver optimization. Firstly, \cite{ref14} and \cite{ref15} proposed a tradeoff technique between PA and MIMO radar, called phased-MIMO radar by partitioning the transmit array into a number of subarrays that are beamformed towards a direction of interest. Focusing on the performance improvement of joint range and angle estimation, \cite{ref16} designed a subaperture scheme for FDA radar and \cite{ref17} discussed the Cram\'{e}r-Rao Bound (CRB) under the assumption of time-variant beam and additive colored noise. Secondly, \cite{ref18} designed a space-time coding approach to improve the range resolution by using the transmit diversity technique. Moreover, \cite{ref19} proposed a method of quadratic phase coding for vertical FDA radar to effectively suppress range-ambiguous clutter. Thirdly, \cite{ref20} addressed a joint transmitter and receiver optimization problem for FDA-MIMO radar to improve detection performance under diverse clutter environments. For target localization, \cite{ref21} proposed a coprime frequency offset strategy and coprime array structure for FDA-MIMO radar to enhance the accuracy of parameter estimation. Furthermore, some approaches focus on the transmit beampattern synthesis, especially range-dependent and time-variance beam for FDA-MIMO radar, to obtain focused main-beam by optimization of frequency offset \cite{ref22}\cite{ref23}\cite{ref24}\cite{ref25}. Nevertheless, these methods are addressed to some specific detection scenarios with the given optimization objectives and not an optimal combination of PA and FDA-MIMO radar in the generalized model.

The main advantages of FDA-MIMO radar can be summarized into two categories: mainlobe interference suppression and range-ambiguous clutter suppression. The first category uses the range-dependency of FDA-MIMO radar to improve the output signal-to-interference-plus-noise ratio (SINR) for target detection when the jamming has the same azimuth as the target but different ranges\cite{ref26}\cite{ref27}. In this case, apart from generating the range-dimensional zero-trapping, range-dependent transmit spatial frequency of FDA-MIMO radar can be exploited to effectively improve the range resolution by the secondary range cell\cite{ref28}, discriminate the mainlobe deceptive jamming\cite{ref29}, and suppress sidelobe deceptive jamming\cite{ref30}. In addition, for the covariance matrix estimation in adaptive filtering, \cite{ref31} and \cite{ref32} proposed the waveform optimization for MIMO radar and \cite{ref33} analyzed the adaptive power allocation for FDA-MIMO radar, emphasizing the benefits of waveform diversity. 

The second category includes range-ambiguous clutter suppression and mainlobe clutter suppression in moving target detection by exploiting the space-time-range adaptive processing (STRAP)\cite{ref25}\cite{ref27}, which is an extension of space-time adaptive processing (STAP)\cite{ref34} for FDA-MIMO radar \cite{ref35}\cite{ref36}. \cite{ref37} proposed a method to separate range-ambiguous clutter and range-unambiguous clutter by using the operation of secondary range dependence compensation (SRDC), which is also systematically summarized in the slow-time technology application for FDA-MIMO radar \cite{ref38}. Moreover, through the joint information of range and angle in transmit spatial frequency, \cite{ref39} and \cite{ref40} contributed to the clutter covariance matrix estimation for FDA-MIMO radar. However, waveform orthogonality for the aforementioned FDA-MIMO radar and signal coherency for PA radar are contradictory from the perspective of designing the frequency offset.

The technique that combines the advantages of PA and FDA-MIMO radar is naturally known as coherent FDA (C-FDA) radar. It takes a smaller frequency offset to simultaneously obtain higher coherent process gain, waveform diversity, and range-dependency. Several approaches have proposed to design C-FDA receivers by using Doppler pulse coding \cite{ref25}, angle-dependent matched filtering \cite{ref41}, and randomly assigned carrier frequency\cite{ref42}, but they do not yield large array gains and the receiver models are not generalizable to all FDA-MIMO radar scenarios because of their different design goals (\cite{ref25} aims to suppress mainlobe clutter while \cite{ref41} concerns moving target indication (MTI), and \cite{ref42} is for active sensing). To develop the C-FDA radar without the limitation on the frequency offset requirements, the following difficulties need to be overcome.
\begin{enumerate}
\item The FDA-MIMO radar receiver cannot be applied to C-FDA radar since each transmit signal spectrum is highly overlapped due to the frequency offset that is much smaller than the bandwidth. 
\item After joint transmitter and receiver processing, C-FDA radar should fulfill a coherent processing gain no less than the PA radar to be more robust against background noise.
\item The C-FDA radar receiver requires multi-channel processing to ensure waveform diversity while keeping range-dependency from the frequency offset, inheriting the capability of mainlobe interference suppression and range-ambiguous clutter suppression from FDA-MIMO radar.
\end{enumerate}

To address these issues, this paper proposes a new approach for the design of the transmitter and receiver of C-FDA radar. Furthermore, focusing on mainlobe interference suppression and range-ambiguous clutter suppression, this paper analyzes the weakness of FDA-MIMO radar due to the requirement on the frequency offset and introduces the benefits of C-FDA radar, where the limitation on frequency offset doesn't apply any more. Overall, the proposed C-FDA radar:
\begin{enumerate}
\item avoids the channel confusion of target echo by implementing multi-channel frequency mixing (MFM) and multi-channel matched filtering (MMF).
\item enjoys higher coherent processing gain than PA radar.
\item has range-dependent transmit spatial frequency as FDA-MIMO radar.
\item resolves secondary range-ambiguity (SRA) caused by a large frequency offset.
\end{enumerate}

The rest of paper is organized as follows. Generalized model preliminaries of PA and FDA-MIMO radar are presented in Section II, where some definitions required by the proposed C-FDA radar are emphasized. Then we introduce the principles of transmitter and receiver for C-FDA radar in Section III. Section IV and Section V present the performance analysis of C-FDA radar in terms of mainlobe interference suppression and range-ambiguous clutter suppression, respectively. In Section VI, we show the numerical results to verify the effectiveness of C-FDA radar. Section VII draws our conclusions at last.

\textbf{Notations}: Vectors and matrices are denoted by bolded lowercase and uppercase, respectively. $(\cdot)^T$, $(\cdot)^H$, $\odot$, $\circledast$, and $\otimes$ denote the transpose, conjugate transpose, Hadamard product, convolution operation, and Kronecker product, respectively. $\boldsymbol{I}_M$ and $\mathbf{1}_M$ stand for the M-dimentional identity matrix, the M-dimentional all 1 vector or matrix. $\mathbb{C}$ and $\mathbb{N}$ are the sets of complex numbers and integers, and $\mathbb{C}^{N\times M}$ means the Euclidean space of ($N\times M$)-dimensional complex matrices (or vectors if $M=1$ or $N=1$). $\mathrm{diag}(x_1,...,x_N)$ represents a diagonal matrix with its diagonal elements of $x_1,...,x_N$. $\lfloor \cdot \rfloor $ denotes a downward rounding operation. $\mathcal{F} \left\{ \cdot \right\}$, $\mathrm{E}\left\{ \cdot \right\}$ denote the Fourier transform operation and the expection operation, respectively. 

\section{PA and FDA-MIMO radar:\\ Generalized model}

Consider a monostatic radar system with a linear array deployed on the aircraft to detect the moving target of interest. The platform is flying with the yaw angle\footnote[1]{The yaw angle means the angle between the flight direction and the positive direction of X-axis \cite{ref35} \cite{ref39}.} $\psi$ and velocity $v_a$ at altitude $H$. $M$ and $N$ denote the number of transmit and receive elements, respectively. The transmit signals of $M$ elements can be defined by a $M$-dimensional vector.
\begin{equation}
\boldsymbol{u}\left( t \right) \triangleq \left[ \begin{matrix}
        u_1\left( t \right)&    \cdots&        u_m\left( t \right)&            \cdots&         u_M\left( t \right)\\
\end{matrix} \right] ^T
\label{eq.1}
\end{equation}
where $u_m\left( t \right)$ denotes the signal transmitted by the m-th transmit element. 
\begin{equation}
u_m\left( t \right) =\phi \left( t \right) \exp \left\{ j2\pi [f_c+\left( m-1 \right) \varDelta f]t \right\}
\label{eq.2}
\end{equation}
And $\phi \left( t \right)$ is the baseband signal with unit energy of time width $T_p$ and bandwidth $B$. $f_c$ is the carrier frequency. $\varDelta f$ denotes the frequency offset, which determines the generalized signal model for PA and FDA-MIMO radar.
\begin{equation}
\int_{T_p}{\boldsymbol{u}\left( t \right) \boldsymbol{u}^H\left( t \right) \mathrm{d}t}=\left\{ \begin{array}{c}
        \mathbf{1}_M,~\varDelta f=0\\
        \boldsymbol{I}_M,~\varDelta f\geqslant B\\
\end{array} \right. 
\label{eq.3}
\end{equation}
where PA radar requires $\varDelta f=0$ and FDA-MIMO radar requires $\varDelta f\geqslant B$ to satisfy the waveform orthogonality (See in Fig.\ref{fig_2a}). 

Assume that the azimuth, elevation, range, velocity of the target of interest are $\varphi_t$, $\theta_t$, $R_t=H/\sin\theta_t$, and $v_t$, respectively. For one fast-time snapshot sampled in the noise-domain environment, the $NK\times 1$ vector received by $N$ receive elements for $K$ pulses on slow-time dimension can be modeled as
\begin{equation}
 \boldsymbol{x}\left( t \right) =\boldsymbol{x}_t\left( t \right) +\boldsymbol{x}_d\left( t \right) +\boldsymbol{n}\left( t \right) 
\label{eq.4}
\end{equation} 
where $\boldsymbol{x}\left( t \right)$, $\boldsymbol{x}_d\left( t \right)$, and $\boldsymbol{n}\left( t \right)$ are the independent components related to the target, disturbance\footnote[2]{Disturbance includes jamming or clutter according to the specific scenario.}, and noise, respectively. Under the assumptions of far-field point target and narrow-band model \cite{ref1}, the target signal can be expressed as
\begin{equation}
\boldsymbol{x}_t\left( t \right) =\xi _t\left[ \boldsymbol{a}_{t}^{T}\left( \varphi _t,\theta _t \right) \boldsymbol{u}\left( t \right) \right] \boldsymbol{a}_r\left( \varphi _t,\theta _t \right) \otimes \boldsymbol{a}_D\left( f_D \right) 
\label{eq.5}
\end{equation}
where $\xi _t$ represents the target scattering coefficient.
\begin{subequations}
\begin{align}
\boldsymbol{a}_t\left( \varphi ,\theta \right) =&\left[ \begin{matrix}
        1&              e^{j2\pi f_{\varphi}}&          \cdots&         e^{j2\pi \left( M-1 \right) f_{\varphi}}\\
\end{matrix} \right] ^T\label{eq.6a}
\\
\boldsymbol{a}_r\left( \varphi ,\theta \right) =&\left[ \begin{matrix}
        1&              e^{j2\pi f_{\varphi}}&          \cdots&         e^{j2\pi \left( N-1 \right) f_{\varphi}}\\
\end{matrix} \right] ^T\label{eq.6b}
\\
\boldsymbol{a}_D\left(f_D \right) =&\left[ \begin{matrix}
        1&              e^{j2\pi f_D}&          \cdots&         e^{j2\pi \left( K-1 \right) f_D}\\
\end{matrix} \right] ^T\label{eq.6c}
\end{align}
\end{subequations}
And $f_{\varphi}$ and $f_D$ can be written as
\begin{subequations}
\begin{align}
f_{\varphi}=&\frac{d}{\lambda}\cos \varphi \cos \theta 
\label{eq.7a}
\\
f_D=&\frac{2\left( v_a+v_t \right) T}{\lambda}\cos \left( \varphi +\psi \right) \cos \theta \label{eq.7b}
\end{align}
\end{subequations}
where $\lambda$, $d$, and $T$ denote the wavelength, array spacing and pulse repetition interval (PRI), respectively. 

After matched filtering the receive signals, the back-scattered signal can be sampled. For PA radar, the $NK\times 1$ target vector can be obtained by the baseband waveforms $\phi \left( t \right)$ \cite{ref1}, 
\begin{align}
\boldsymbol{t}_{\mathrm{PA}}=\int_{T_p}{\boldsymbol{x}_t\left( t \right) \phi ^*\left( t \right) \mathrm{d}t}=\sqrt{M}\xi _t\boldsymbol{a}_r\left( \varphi _t,\theta _t \right) \otimes \boldsymbol{a}_D\left(f_D\right) 
\label{eq.8}
\end{align}
For FDA-MIMO radar, the $MNK\times 1$ target vector can be recovered by the waveforms $\phi \left( t \right)$ after the frequency mixing operation \cite{ref5}\cite{ref17}\cite{ref36}.
\begin{align}
\boldsymbol{t}_{\mathrm{F-M}}=&\int_{T_p}{\left[ \boldsymbol{d}\left( t \right) \otimes \boldsymbol{x}_t\left( t \right) \right] \phi ^*\left( t \right) \mathrm{d}t}
\nonumber\\
=\xi _t&\left[ \boldsymbol{a}_t\left( \varphi _t,\theta _t \right) \odot \boldsymbol{a}_R\left( R_t \right) \right] \otimes \boldsymbol{a}_r\left( \varphi _t,\theta _t \right) \otimes \boldsymbol{a}_D\left( f_D \right) 
\label{eq.9}
\end{align}
where $\boldsymbol{d}\left( t \right)$ and $\boldsymbol{a}_R\left( R \right)$ denote the frequency mixing vector and transmit range-dependent vector, respectively.
\begin{subequations}
\begin{align}
\boldsymbol{d}\left( t \right) =&\left[ \begin{matrix}
        1&              e^{-j2\pi \varDelta ft}&                \cdots&         e^{-j2\pi \left( M-1 \right) \varDelta ft}\\
\end{matrix} \right] ^T
\label{eq.10a}
\\
\boldsymbol{a}_R\left( R \right) =&\left[ \begin{matrix}
        1&              e^{j2\pi f_R}&          \cdots&         e^{j2\pi \left( M-1 \right) f_R}\\
\end{matrix} \right] ^T\label{eq.10b}
\end{align}
\end{subequations}
Particularly, $f_R=-2R\varDelta f/c$ is the range-dependent frequency related to the frequency offset $\varDelta f$ \cite{ref7}\cite{ref28}.

Furthermore, the coherent array gain\footnote[3]{The array gain reflects the improvement in SNR obtained by using the array. It is defined to be the ratio of the SNR at the output of the array to the input SNR \cite{ref1}\cite{ref14}.} can be defined as the improvement in SNR by using the receive beamforming vector $\boldsymbol{w}_{\mathrm{PA}}=\boldsymbol{t}_{\mathrm{PA}}$ and $\boldsymbol{w}_{\mathrm{F-M}}=\boldsymbol{t}_{\mathrm{F-M}}$. 
\begin{subequations}
\begin{align}
\varOmega _{\mathrm{PA}}=&\frac{\sigma _{{t}}^{2}\left| \boldsymbol{w}_{\mathrm{PA}}^{H}\boldsymbol{t}_{\mathrm{PA}} \right|^2}{\boldsymbol{w}_{\mathrm{PA}}^{H}\boldsymbol{R}_{\mathrm{n}}^{\left( \mathrm{PA} \right)}\boldsymbol{w}_\mathrm{PA}}=M^2NK\cdot\mathrm{SNR}_{\mathrm{in}}
\label{eq.11a}
\\
\varOmega _{\mathrm{F}-\mathrm{M}}=&\frac{\sigma _{{t}}^{2}\left| \boldsymbol{w}_{\mathrm{F}-\mathrm{M}}^{H}\boldsymbol{t}_{\mathrm{F}-\mathrm{M}} \right|^2}{\boldsymbol{w}_{\mathrm{F-M}}^{H}\boldsymbol{R}_{\mathrm{n}}^{\left( \mathrm{F}-\mathrm{M} \right)}\boldsymbol{w}_\mathrm{F-M}}=MNK\cdot\mathrm{SNR}_{\mathrm{in}}
\label{eq.11b}
\end{align}
\end{subequations}
where $\sigma _{t}^{2}=\mathrm{E}\left\{ \left| \xi _t \right|^2 \right\} 
$. $\boldsymbol{R}_{\mathrm{n}}^{\left( \mathrm{PA} \right)}=\sigma _{\mathrm{n}}^{2}\boldsymbol{I}_{NK}$ and $\boldsymbol{R}_{\mathrm{n}}^{\left( \mathrm{F}-\mathrm{M} \right)}=\sigma _{\mathrm{n}}^{2}\boldsymbol{I}_{MNK}$ are the noise covariance matrix with noise power $\sigma_{\mathrm{n}}^{2}$ for PA and FDA-MIMO radar, respectively. $\mathrm{SNR}_{\mathrm{in}}=\sigma _{t}^{2}/\sigma _{\mathrm{n}}^{2}$ means the input SNR. Comparing \eqref{eq.11a} and \eqref{eq.11b}, 
\begin{equation}
\varOmega _{\mathrm{PA}}=M\cdot\varOmega _{\mathrm{F-M}}
\label{eq.12}
\end{equation}
which indicates that PA radar is more robust to background noise than the FDA-MIMO radar\cite{ref14}.


\section{Coherent FDA radar: \\Transmitter \& receiver}

We propose a new design approach for the transmitter and receiver of a C-FDA radar to obtain a high coherent gain while maintaining the range-dependency. Note that the proposed radar utilizes a frequency offset much smaller than the bandwidth $B$ of transmit signal to achieve coherent gain, removing the limitation of frequency offset $\varDelta f\geqslant B$ of FDA-MIMO radar. In this section, we introduce the processing principles of the proposed transmitter and receiver through two theorems.

\begin{figure}[!t]
\centering
\includegraphics[width=3.45in]{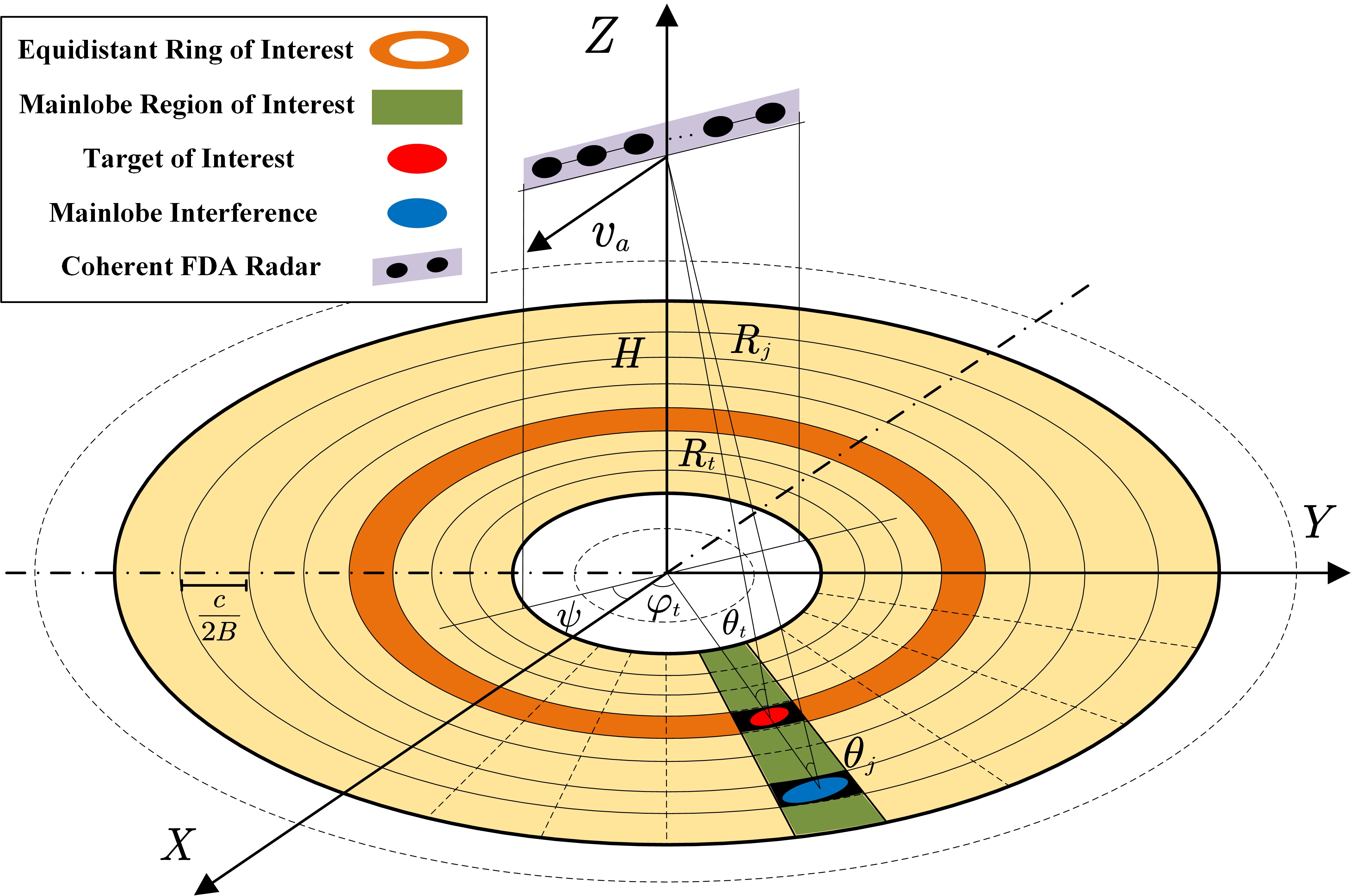}
\caption{Detection scene of C-FDA radar.}
\label{fig_1}
\end{figure}

Consider an airborne C-FDA radar with the same platform and target parameters as the aforementioned PA and FDA-MIMO radar in Fig.\ref{fig_1}. Focusing on the k-th pulse signal from the m-th transmit element to the n-th receive element, the time delay consists of the two-way propagation delay $\tau ^{\left( p \right)}$, the transmit array displacement delay $\tau _{m}^{\left( t \right)}$, the receive array displacement delay $\tau _{n}^{\left( r \right)}$, and the Doppler delay $\tau _{k}^{\left( D \right)}$.
\begin{subequations}
\begin{align}
\tau ^{\left( p \right)}=&\frac{2R_t}{c}\label{eq.13a}
\\
\tau _{m}^{\left( t \right)}=&\frac{d}{c}\left( m-1 \right) \cos \varphi_t \cos \theta_t\label{eq.13b}
\\
\tau _{n}^{\left( r \right)}=&\frac{d}{c}\left( n-1 \right) \cos \varphi_t \cos \theta_t\label{eq.13c}
\\
\tau _{k}^{\left( D \right)}=&\frac{2\left( v_a+v_t \right)}{c}T\left( k-1 \right) \cos \left( \varphi_t +\psi \right) \cos \theta_t\label{eq.13d}
\end{align}
\end{subequations}

\subsection{Transmitter}

Assume the transmit signal vector of C-FDA radar can be expressed as \eqref{eq.1} and \eqref{eq.2}. Then we design a transmit beamforming vector $\boldsymbol{v}$ to align the main-beam toward the target of interest.
\begin{equation}
\boldsymbol{v}=\left[ \begin{matrix}
        v_1&      \cdots&      v_m&            \cdots&         v_M\\
\end{matrix} \right] ^T
\label{eq.14}
\end{equation}
where $v_m$ denotes the beamforming weight for the m-th transmit element.
\begin{equation}
v_m=\exp\left\{-j2\pi \left[ f_c+\left( m-1 \right) \varDelta f \right] \left( m-1 \right) \frac{d}{c}\cos \varphi \cos \theta\right\}
\label{eq.15}
\end{equation}
Then the transmit signal of C-FDA radar after beamforming operation can be expressed as
\begin{align}
s\left( t \right) =&\boldsymbol{v}^{\mathrm{T}}\left[ \begin{matrix}
        u_1\left( t-\tau _{1}^{\left( t \right)} \right)&       \cdots&         u_M\left( t-\tau _{M}^{\left( t \right)} \right)\\
\end{matrix} \right] ^T
\nonumber\\
=&\sum_{m=1}^M{v_{m}^{*}}u\left( t \right) e^{j2\pi \left[ f_c+\left( m-1 \right) \varDelta f \right] \left( t-\tau _{m}^{\left( t \right)} \right)}
\label{eq.16}
\end{align}
Note that \eqref{eq.16} utilizes the narrow-band assumption.

\begin{figure*}[!t]
\centering
\subfloat[Transmit signal spectrum (four transmit elements)]{\includegraphics[width=7.0in]{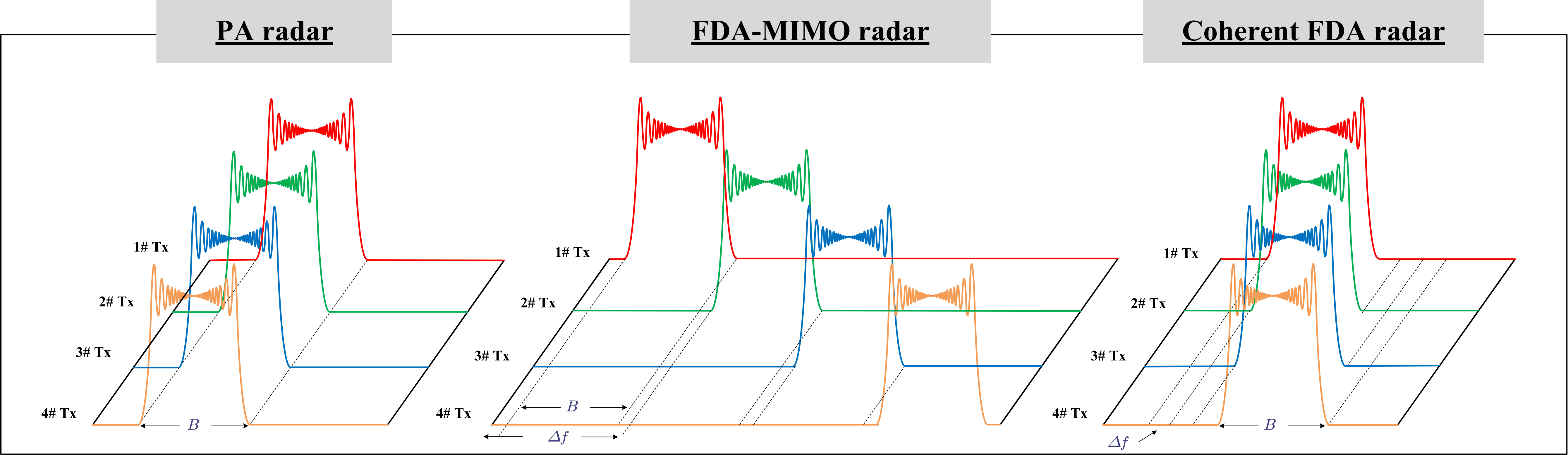}%
\label{fig_2a}}
\hfil
\\
\subfloat[Receiver processing (four channels)]{\includegraphics[width=7.0in]{C_FDArecieve_new.jpg}%
\label{fig_second_case}}
\caption{Transmit signal spectrum and receiver processing (MFM and MMF) for PA, FDA-MIMO, and C-FDA radar. (a) Transmit signal spectrum. (b) Receiver processing.}
\label{fig_2b}
\end{figure*}

\subsection{Receiver}

The back-scattered signal of the k-th pulse received by the n-th receive element can be expressed as
\begin{align}
y_{n,k}\left( t \right) =&\xi _ts\left( t-\tau _{n,k} \right)
\nonumber\\
=&\xi_t \sum_{m=1}^M{v_{m}^{*}}u\left( t \right) e^{j2\pi \left[ f_c+\left( m-1 \right) \varDelta f \right] \left( t-\tau _{m}^{\left( T \right)}-\tau _{n,k} \right)}\label{eq.17}
\end{align}
where $\tau _{n,k}=\tau ^{\left( p \right)}-\tau _{n}^{\left( r \right)}-\tau _{k}^{\left( D \right)}$ contains the information of receive spatial frequency and the Doppler frequency. The target information held in the phase of \eqref{eq.17} can be demodulated by using a special receiver processing. The proposed receiver for C-FDA radar has two essential steps, called the multi-channel frequency mixing (MFM) and the multi-channel matched filtering (MMF). For the MFM, the frequency mixing term in the m-th channel is $e^{-j2\pi \left[ f_c+\left( m-1 \right) \varDelta f \right] t}$. For the next MMF, the matched filtering function for the m-th channel is designed as 
\begin{equation}
h_m\left( t \right) =\phi^*\left( -t \right) \sum_{i=1}^M{e^{-j2\pi \left( m-i \right) \varDelta ft}}
\label{eq.18}
\end{equation}
Note that the receive channel corresponds to the transmit element due to the waveform diversity. Fig.2 explicitly shows the procedure of the proposed receiver with four channels ($M=4$) as an example.

After the proposed transmitter and receiver processing, the signal data of $K$ pulses received by $N$ elements (each receive element has $M$ channels) can be expressed as a $MNK\times 1$ space-time-range vector for one fast-time snapshot.
\begin{equation}
\boldsymbol{t}_{\mathrm{C-F}}=E\xi_t \boldsymbol{a}_R\left( R_t \right)\otimes \boldsymbol{a}_{\varphi_t}\left( \varphi_t,\theta_t \right)\otimes \boldsymbol{a}_D\left( f_D \right)
\label{eq.19}
\end{equation}
where $\boldsymbol{a}_R\left( R_t \right)$, $\boldsymbol{a}_{\varphi}\left( \varphi_t,\theta_t \right)$, and $\boldsymbol{a}_D\left( f_D \right)$ represent the transmit range-dependent vector, the receive spatial steering vector, and the Doppler vector for C-FDA radar. $E$ denotes the amplitude coefficient of the sampling snapshot associated with the frequency offset $\varDelta f$. We elaborate these vectors and scalar by proposing the following two theorems while using their proofs to mathematically fomulate the signal processing of the proposed receiver for C-FDA radar.

\begin{thm}[Phase analysis for receive processing]\label{Thm.1}
~\; ~Through the MFM and MMF of the proposed receiver in $M\times N$ receive channels, all channels of each receive element can take a range-dimensional peak at $R_\mathrm{peak}\approx R_t$ that uniquely corresponds to the back-scattered signal from the target of interest. Furthermore, the transmit range-dependent vector $\boldsymbol{a}_R\left( R_t \right)$, the receive spatial steering vector $\boldsymbol{a}_{\varphi}\left( \varphi_t,\theta_t \right)$, and the Doppler vector $\boldsymbol{a}_D\left( f_D \right)$ of the $MNK\times 1$ sampling data can be written as 
\begin{subequations}
\begin{align}
\boldsymbol{a}_R\left( R \right)=&\left[ \begin{matrix}
        1&              e^{j2\pi f_R}&          \cdots&         e^{j2\pi \left( M-1 \right) f_R}\\
\end{matrix} \right] ^{{T}}\label{eq.20a}
\\
\boldsymbol{a}_{\varphi}\left( \varphi,\theta \right)=&\left[ \begin{matrix}
        1&              e^{j2\pi f_{\varphi}}&          \cdots&         e^{j2\pi \left( N-1 \right) f_{\varphi}}\\
\end{matrix} \right] ^{{T}}\label{eq.20b}
\\
\boldsymbol{a}_D\left( f_D \right)=&\left[ \begin{matrix}
        1&              e^{j2\pi f_D}&          \cdots&         e^{j2\pi \left( K-1 \right) f_D}\\
\end{matrix} \right] ^{{T}}\label{eq.20c}
\end{align}
\end{subequations} 
\end{thm}

\begin{proof}[Proof of Theorem \ref{Thm.1}]
See Appendix A.
\end{proof}

\begin{thm}[Amplitude analysis for receive processing]\label{Thm.2}
~\; Assume each transmit element emits a unit energy signal with a time width of $T_p$ and a bandwidth of $B$. When the frequency offset of C-FDA radar is $0\leqslant \varDelta f\ll B/\left( M-1 \right) $, then the amplitude coefficient $E$ of the $MNK\times 1$ samples can be restricted as follow:
\begin{equation}
\sqrt{M}<E\leqslant M
\label{eq.21}
\end{equation}
where the right-side equal symbol holds if and only if $\varDelta f= 0$. Furthermore, the comparison of the array gains for PA radar, FDA-MIMO radar, and C-FDA radar can be expressed as 
\begin{equation}
M\varOmega _{\mathrm{F}-\mathrm{M}}=\varOmega _{\mathrm{PA}}<\varOmega _{\mathrm{C}-\mathrm{F}}\leqslant M^2\varOmega _{\mathrm{PA}}
\label{eq.22}
\end{equation}
\end{thm}
\begin{proof}[Proof of Theorem \ref{Thm.2}]
See Appendix B.
\end{proof} 

Through the proposed receiver, we can obtain the space-time-range signal model of C-FDA radar with high array gain and range-dependent transmit frequency. Based on these two properties, we discuss the use of C-FDA radar for mainlobe interference suppression and range-ambiguous clutter suppression and we show that it has significant advantages over PA and FDA-MIMO radar.

\section{Mainlobe interference suppression: \\SRA problem and output SINR}

FDA-based radar, that has the transmit spatial frequency with range-dependency, has benefits in mainlobe interference suppression, which has been highlighted in \cite{ref26}\cite{ref27}\cite{ref30}. In this section, we focus on the SRA problem that deteriorates the performance of mainlobe interference suppression when $\varDelta f\geqslant B$ for FDA-MIMO radar and we show the advantages of using C-FDA radar without frequency offset constraints.

Consider a towed mainlobe jamming with the same azimuth and Doppler as the target but at different ranges, denoted by $R_j$, as shown in Fig.\ref{fig_1}. For FDA-MIMO radar and C-FDA radar, the received jamming data vector can be expressed as
\begin{subequations}
\begin{align}
\boldsymbol{j}_{\mathrm{F}-\mathrm{M}}=&\xi _j\left[ \boldsymbol{a}_t\left( \varphi _t,\theta _t \right) \odot \boldsymbol{a}_R\left( R_j \right) \right] \otimes \boldsymbol{a}_r\left( \varphi _t,\theta _t \right) \otimes \boldsymbol{a}_D\left( f_D \right) 
\label{eq.23a}
\\
\boldsymbol{j}_{\mathrm{C}-\mathrm{F}}=&E\xi _j\boldsymbol{a}_R\left( R_t \right) \otimes \boldsymbol{a}_{\varphi}\left( \varphi _t,\theta _t \right) \otimes \boldsymbol{a}_D\left( f_D \right) 
\label{eq.23b}
\end{align}
\end{subequations}
where $\xi_j$ denotes the interference scattering coefficient. Based on the minimum variance distortionless response (MVDR) beamformer, the range-dimensional jamming power can be minimized by the optimal beamforming weight vector $\boldsymbol{w}^{(j)}_{\mathrm{opt}}$.
\begin{equation}
\boldsymbol{w}^{(j)}_{\mathrm{opt}}=\eta \boldsymbol{R}_{\mathrm{j}+\mathrm{n}}^{-1}\boldsymbol{t}
\label{eq.24}
\end{equation}
where $\eta =1/\left( \boldsymbol{t}^H\boldsymbol{R}_{\mathrm{j}+\mathrm{n}}^{-1}\boldsymbol{t} \right)$. $\boldsymbol{t}$ can be replaced by $\boldsymbol{t}_{\mathrm{F-M}}$ in \eqref{eq.9} and $\boldsymbol{t}_{\mathrm{C-F}}$ in \eqref{eq.19} corresponding to FDA-MIMO radar and C-FDA radar, respectively. Similarly, $\boldsymbol{R}_{\mathrm{j}+\mathrm{n}}$ is the jamming plus noise covariance matrix that can be replaced by $\boldsymbol{R}_{\mathrm{j}}^{\left( \mathrm{F}-\mathrm{M} \right)}+\boldsymbol{R}_{\mathrm{n}}^{\left( \mathrm{F}-\mathrm{M} \right)}$ and $\boldsymbol{R}_{\mathrm{j}}^{\left( \mathrm{C}-\mathrm{F} \right)}+\boldsymbol{R}_{\mathrm{n}}^{\left( \mathrm{C}-\mathrm{F} \right)}$, where $\boldsymbol{R}_{\mathrm{j}}^{\left( \mathrm{F}-\mathrm{M} \right)}=\mathrm{E}\{ \boldsymbol{j}_{\mathrm{F}-\mathrm{M}}\boldsymbol{j}_{\mathrm{F}-\mathrm{M}}^{H} \} $ and $\boldsymbol{R}_{\mathrm{j}}^{\left( \mathrm{C}-\mathrm{F} \right)}=\mathrm{E}\{ \boldsymbol{j}_{\mathrm{C}-\mathrm{F}}\boldsymbol{j}_{\mathrm{C}-\mathrm{F}}^{H} \} $. 

By using the matrix inversion lemma\cite{ref43}, the output SINR of the C-FDA radar can be expressed as
\begin{align}
\mathrm{SINR}^{\mathrm{C-F}}_{\mathrm{o}}=&\frac{\left| \left[ \boldsymbol{w}_{\mathrm{opt}}^{(j)} \right] ^H\boldsymbol{t} \right|^2}{\left[ \boldsymbol{w}_{\mathrm{opt}}^{(j)} \right] ^H\boldsymbol{R}_{\mathrm{j}+\mathrm{n}}\boldsymbol{w}_{\mathrm{opt}}^{(j)}}
\nonumber\\
=&\frac{1}{\sigma _{\mathrm{n}}^{2}}\left[ \left| \boldsymbol{t} \right|^2-\frac{\sigma _{\mathrm{j}}^{2}}{\sigma _{\mathrm{n}}^{2}+MNK\sigma _{\mathrm{j}}^{2}}\left| \boldsymbol{t}^H\boldsymbol{j} \right|^2 \right] 
\nonumber\\
=\mathrm{SNR}_{\mathrm{in}}&\left[ E^2MNK-\frac{E^4N^2K^2\cdot \mathrm{INR}  }{1+E^2MNK\cdot \mathrm{INR}}\left| \varPhi _R \right|^2 \right] \label{eq.25}
\end{align}
where $\mathrm{SNR}_{\mathrm{in}}=\sigma _{\mathrm{t}}^{2}/\sigma _{\mathrm{n}}^{2}$, $\sigma _{\mathrm{t}}^{2}=\mathrm{E}\{ \left| \xi_t \right|^2 \}$, $\mathrm{INR}=\sigma _{\mathrm{j}}^{2}/\sigma _{\mathrm{n}}^{2}$, and $\sigma _{\mathrm{j}}^{2}=\mathrm{E}\{ \left| \xi_j \right|^2 \}$. For comparison, the output SINR of the FDA-MIMO radar can be written as
\begin{equation}
\mathrm{SINR}^{\mathrm{F-M}}_{\mathrm{o}}=\mathrm{SNR}_{\mathrm{in}}\left[ MNK-\frac{N^2K^2\cdot \mathrm{INR}  }{1+MNK\cdot \mathrm{INR}}\left| \varPhi _R \right|^2 \right] \label{eq.26}
\end{equation}
And $\varPhi _R$ in \eqref{eq.25} and \eqref{eq.26} can be expressed as
\begin{align}
\varPhi _{R}=\frac{\sin \left[ M\pi\left( 2 \varDelta R \frac{\varDelta f}{c} \right) \right]}{\sin \left[\pi\left( 2\varDelta R \frac{\varDelta f}{c} \right)\right]} 
\label{eq.27}
\end{align}
where $\varDelta R=\left| R_j-R_t \right|$. From \eqref{eq.25} and \eqref{eq.26}, the output SINR is minimum when $\varPhi _{R}$ is maximum, which reveals the relationship between $\varDelta f$ and $\varDelta R$ and point out the SRA problem. According to \eqref{eq.27}, the relationship between the maximum of $\varPhi _{R}$ and $\varDelta R$ can be expressed as
\begin{equation}
\mathop {\mathrm{arg}\max} \limits_{\varDelta R}\varPhi _R=\frac{c}{2\varDelta f}L, L\in \mathbb{N} 
\label{eq.28}
\end{equation}
where $c/(2\varDelta f)$ is called the secondary range ambiguity period and $L$ is the secondary range ambiguity number. We assume that the detection region is $\left[ R_t-{R_d}/{2},R_t+{R_d}/{2} \right] $ and the maximum unambiguous range is $R_d=c/(2T)$. When the target range $R_t$ and $R_d$ is fixed, \eqref{eq.28} indicates that the jamming range satisfying \eqref{eq.supp} can severely deteriorate the performance of mainlobe interference suppression.
\begin{equation}
R_j=R_t\pm \frac{c}{2\varDelta f}L,\,\,L\in \mathbb{N}
\label{eq.supp}
\end{equation}
Note that the number of the jamming ranges satisfying \eqref{eq.supp} is $J=\lfloor \frac{\varDelta f}{T} \rfloor $. For FDA-MIMO radar, a large $\varDelta f$ results in a large $J$, which means the mainlobe interference suppression is invalid when the jamming is located at these ranges. The performance of mainlobe interference suppression for FDA-MIMO radar significantly depends on the jamming distance relative to the target. For the proposed C-FDA radar, a much smaller frequency offset $\varDelta f$ yields a very small $J$, which means the performance of mainlobe interference suppression for C-FDA radar is much less dependent on the jamming distance relative to the target. In Section VI, we describe some numerical results.

\section{Range ambiguity clutter suppression: Strategy and Performance}

FDA-MIMO radar has advantages in range-ambiguous clutter suppression due to the range-dependency of transmit spatial frequency\cite{ref37}\cite{ref38}. In this section, we discuss the limitations encountered by FDA-MIMO radar in clutter suppression and address them by using the proposed C-FDA radar.

Consider an airborne radar disturbed by ground clutter when detecting a moving target within the range-unambiguous region, as shown in Fig.\ref{fig_1}. One fast-time sampled signal contains the range-unambiguous clutter from the cell under test (CUT) and the remote range-ambiguous clutter. Assume the ambiguity range is $R_u=cT/2$ and the total ambiguity number is $P$, where the ambiguity range related to the ambiguity number of $p$ is $R_p=R_t+pR_u,~p=1...P$. The clutter data can be expressed as a matrix $\boldsymbol{V}_c$. 
\begin{equation}
\boldsymbol{V}_c=\left[ \begin{array}{ccccc}
        \boldsymbol{V}_u&    \vline  &         \boldsymbol{V}_{a}^{\left( 1 \right)}&          \cdots&         \boldsymbol{V}_{a}^{\left( P \right)}\\
\end{array} \right]
\label{eq.29}
\end{equation}
where $\boldsymbol{V}_u$ and $\boldsymbol{V}_{a}^{\left( p \right)},~p=1...P$ are the range-unambiguous and range-ambiguous clutter, respectively. 
\begin{subequations}
\begin{align}
\boldsymbol{V}_u=&\left[ \begin{matrix}
        \xi_1\boldsymbol{c}_1&           \cdots&      \xi_i\boldsymbol{c}_i&               \cdots&         \xi_I\boldsymbol{c}_I\\
\end{matrix} \right] 
\label{eq.30a}
\\
\boldsymbol{V}_{a}^{\left( p \right)}=&\left[ \begin{matrix}
       \xi_1^{(p)} \boldsymbol{v}_{1}^{\left( p \right)}&     \cdots&       \xi_i^{(p)}\boldsymbol{v}_{i}^{\left( p \right)}&          \cdots&        \xi_I^{(p)} \boldsymbol{v}_{I}^{\left( p \right)}\\
\end{matrix} \right] 
\label{eq.30b}
\end{align}
\end{subequations}
where $I$ denotes the number of clutter pitches with different azimuth in one equidistant clutter ring. $\xi_i$ and $\xi_i^{(p)}$ denotes the scattering coefficient for range-unambiguous and range-ambiguous clutter pitch. For C-FDA radar, $\boldsymbol{c}_i$ and $\boldsymbol{v}_{i}^{\left( p \right)}$ can be expressed as
\begin{subequations}
\begin{align}
\boldsymbol{c}_i=&E\boldsymbol{a}_R\left( R_t \right) \otimes \boldsymbol{a}_{\varphi}\left( \varphi _i,\theta _t \right) \otimes \boldsymbol{a}_D\left( f_D^{(i)} \right)  \label{eq.31a}
\\
\boldsymbol{v}_{i}^{\left( p \right)}=&E\boldsymbol{a}_R\left( R_p \right) \otimes \boldsymbol{a}_{\varphi}\left( \varphi _i,\theta _p \right) \otimes \boldsymbol{a}_D\left( f_D^{(i)} \right)  \label{eq.31b}
\end{align}
\end{subequations}
For FDA-MIMO radar, then
\begin{subequations}
\begin{align}
\boldsymbol{c}_i=&\left[ \boldsymbol{a}_t\left( \varphi _i,\theta _t \right) \odot \boldsymbol{a}_R\left( R_t \right) \right] \otimes \boldsymbol{a}_r\left( \varphi _i,\theta _t \right) \otimes \boldsymbol{a}_D\left( f_D^{(i)} \right) 
\label{eq.32a}
\\
\boldsymbol{v}_{i}^{\left( p \right)}=&\left[ \boldsymbol{a}_t\left( \varphi _i,\theta _p \right) \odot \boldsymbol{a}_R\left( R_p \right) \right] \otimes \boldsymbol{a}_r\left( \varphi _i,\theta _p \right) \otimes \boldsymbol{a}_D\left( f_D^{(i)} \right) \label{eq.32b}
\end{align}
\end{subequations}
where $f_D^{(i)}=\frac{2v_aT}{\lambda}\cos \theta _t\cos \left( \varphi _i+\psi \right)$ and $\theta_p=\sin (H/R_p)$. 

Typically, the SRDC is used as pre-processing before range-ambiguous clutter suppression for FDA-MIMO radar\cite{ref37}. Constructing the compensating vectors $\boldsymbol{r}_{c}$ respect to the range of CUT, $\boldsymbol{r}_c=\boldsymbol{r}\otimes \mathbf{1}_{N}\otimes \mathbf{1}_{K}$, where
\begin{equation}
\boldsymbol{r}=\left[ \begin{matrix}  1&  e^{j2\pi \frac{2R_t\varDelta f}{c}}&  \cdots&    e^{j2\pi \left( M-1 \right) \frac{2R_t\varDelta f}{c}}\\
\end{matrix} \right] ^T
\label{eq.33}
\end{equation}
then the clutter covariance matrix can be written as
\begin{align}
\boldsymbol{R}_c=&\mathrm{E}\left\{ \boldsymbol{V}_c\boldsymbol{V}_{c}^{H} \right\} =\boldsymbol{U}_c\boldsymbol{\varSigma }_c\boldsymbol{U}_{c}^{H}
\nonumber\\
=&\boldsymbol{U}_u\boldsymbol{\varLambda }_c\boldsymbol{U}_{u}^{H}+\sum_{p=1}^P{\boldsymbol{U}_{a}^{\left( p \right)}\boldsymbol{\varLambda }_{a}^{\left( p \right)}\left( \boldsymbol{U}_{a}^{\left( p \right)} \right) ^H}
\label{eq.34}
\end{align}
where we define $\sigma _{i}^{2}=\mathrm{E}\left\{ \left| \xi _{i} \right|^2 \right\} $ and $(\sigma _{i}^{(p)})^{2}=\mathrm{E}\left\{ \left| \xi _{i}^{(p)} \right|^2 \right\} $. Then,
\begin{subequations}
\begin{align}
\boldsymbol{U}_c=&\left[ \begin{array}{ccccc}
        \boldsymbol{U}_u&    \vline  &         \boldsymbol{U}_{a}^{\left( 1 \right)}&          \cdots&         \boldsymbol{U}_{a}^{\left( P \right)}\\
\end{array} \right]
\label{eq.35a}
\\
\boldsymbol{\varSigma }_c=&\mathrm{diag}\left\{ \begin{array}{cccccc}
        \sigma _{1}^{2}&     \cdots&   \vline  &  ( \sigma _{1}^{\left( 1 \right)}) ^2& \cdots&      ( \sigma _{I}^{\left( P \right)}) ^2\\
\end{array} \right\}\label{eq.35b}
\\
\boldsymbol{U}_u=&\left[ \begin{matrix}
        \boldsymbol{c}_1\odot \boldsymbol{r}_c&               \boldsymbol{c}_2\odot \boldsymbol{r}_c&               \cdots&         \boldsymbol{c}_I\odot \boldsymbol{r}_c\\
\end{matrix} \right] \label{eq.35c}
\\
\boldsymbol{U}_{a}^{\left( p \right)}=&\left[ \begin{matrix}
        \boldsymbol{v}_{1}^{\left( p \right)}\odot \boldsymbol{r}_c&          \boldsymbol{v}_{2}^{\left( p \right)}\odot \boldsymbol{r}_c&          \cdots&         \boldsymbol{v}_{I}^{\left( p \right)}\odot \boldsymbol{r}_c\\
\end{matrix} \right] \label{eq.35d}
\\
\boldsymbol{\varLambda }_c=&\mathrm{diag}\left\{ \begin{matrix}
        \sigma _{1}^{2}&        \cdots&   \sigma _{i}^{2}&  \cdots&       \sigma _{I}^{2}\\
\end{matrix} \right\} \label{eq.35e}
\\
\boldsymbol{\varLambda }_{a}^{\left( p \right)}=&\mathrm{diag}\left\{ \begin{matrix}
        ( \sigma _{1}^{\left( p \right)} ) ^2&               \cdots&   ( \sigma _{i}^{\left( p \right)} ) ^2  \cdots&   ( \sigma _{I}^{\left( p \right)} ) ^2\\
\end{matrix} \right\} \label{eq.35f}
\end{align}
\end{subequations}
Note that $\boldsymbol{U}_c \in \mathbb{C} ^{MNK\times W}$ and $\boldsymbol{\varSigma }_c\in \mathbb{C} ^{W\times W}$, where $W =I (P+1)$. Thereby, the disturbance covariance matrix can be written as 
\begin{equation}
\boldsymbol{R}_d=\boldsymbol{R}_c+\sigma _{\mathrm{n}}^{2}\boldsymbol{I}_{MNK}
\label{eq.36}
\end{equation}
Then the optimal weight vector of STRAP filter can be calculated by the inversion of disturbance covariance matrix,
\begin{equation}
\boldsymbol{w}^{(c)}_{\mathrm{opt}}=\eta_c \boldsymbol{R}_{d}^{-1}\boldsymbol{t}
\label{eq.37}
\end{equation}
where $\eta_c=1/\left( \boldsymbol{t}^H\boldsymbol{R}_{d}^{-1}\boldsymbol{t} \right)$. The output signal-to-disturbance (clutter plus noise) ratio (SDR) can be calculated as
\begin{align}
\mathrm{SDR}_{\mathrm{o}}&=\frac{\left| \left[ \boldsymbol{w}_{\mathrm{opt}}^{(c)} \right] ^H\boldsymbol{t} \right|^2}{\left[ \boldsymbol{w}_{\mathrm{opt}}^{(c)} \right] ^H\boldsymbol{R}_d\boldsymbol{w}_{\mathrm{opt}}^{(c)}}
\nonumber\\
=\frac{1}{\sigma _{\mathrm{n}}^{2}}&\left[ \left| \boldsymbol{t} \right|^2-\boldsymbol{t}^H\boldsymbol{U}_c\left( \boldsymbol{I}_W+\frac{1}{\sigma _{\mathrm{n}}^{2}}\boldsymbol{\varSigma }_c\boldsymbol{U}_{c}^{H}\boldsymbol{U}_c \right) ^{-1}\boldsymbol{\varSigma }_c\boldsymbol{U}_{c}^{H}\boldsymbol{t} \right] 
\nonumber\\
\approx \frac{1}{\sigma _{\mathrm{n}}^{2}}&\left[ \left| \boldsymbol{t} \right|^2-\boldsymbol{t}^H\left[ \boldsymbol{U}_u\boldsymbol{\varGamma }_u\boldsymbol{U}_{u}^{H}+\sum_{p=1}^P{\boldsymbol{U}_{a}^{\left( p \right)}\boldsymbol{\varGamma }_{a}^{\left( p \right)}\left( \boldsymbol{U}_{a}^{\left( p \right)} \right) ^H} \right] \boldsymbol{t} \right] 
\label{eq.38}
\end{align}
where 
\begin{subequations}
\begin{align}
\boldsymbol{\varGamma }_u=&\mathrm{diag}\left\{ \begin{matrix}
        \varrho _1&     \cdots&   \varrho _i& \cdots&     \varrho _I\\
\end{matrix} \right\} 
\label{eq.39a}
\\
\boldsymbol{\varGamma }_{a}^{\left( p \right)}=&\mathrm{diag}\left\{ \begin{matrix}
        \varsigma _{1}^{\left( p \right)}&    \cdots&   \varsigma _{i}^{\left( p \right)}&    \cdots&   \varsigma _{I}^{\left( p \right)}\\
\end{matrix} \right\}  \label{eq.39b}
\end{align}
\end{subequations}
For C-FDA radar, $\varrho _i$ and $\varsigma _{i}^{\left( p \right)}$ can be expressed as
\begin{subequations}
\begin{align}
\varrho _i=&\frac{E^4\sigma _{\mathrm{n}}^{2}\sigma _{i}^{2}}{E^2MNK\sigma _{i}^{2}+\sigma _{\mathrm{n}}^{2}}
\label{eq.40a}
\\
\varsigma _{i}^{\left( p \right)}=&\frac{E^4\sigma _{\mathrm{n}}^{2}( \sigma _{i}^{\left( p \right)}) ^2}{E^2MNK( \sigma _{i}^{\left( p \right)}) ^2+\sigma _{\mathrm{n}}^{2}} \label{eq.40b}
\end{align}
\end{subequations}
For FDA-MIMO radar, then
\begin{subequations}
\begin{align}
\varrho _i=&\frac{\sigma _{\mathrm{n}}^{2}\sigma _{i}^{2}}{MNK\sigma _{i}^{2}+\sigma _{\mathrm{n}}^{2}}
\label{eq.41a}
\\
\varsigma _{i}^{\left( p \right)}=&\frac{\sigma _{\mathrm{n}}^{2}( \sigma _{i}^{\left( p \right)}) ^2}{MNK( \sigma _{i}^{\left( p \right)}) ^2+\sigma _{\mathrm{n}}^{2}} \label{eq.41b}
\end{align}
\end{subequations}
Thereby, the output SDR of C-FDA and FDA-MIMO radar can be expressed as 
\begin{subequations}
\begin{align}
\mathrm{SDR}_{\mathrm{o}}^{\left( \mathrm{C}-\mathrm{F} \right)}=&\mathrm{SNR}_{\mathrm{in}} \left[E^2MNK - \sum_{i=1}^I{\varrho _iM^2\left| \varPhi _{i}^{\left( \varphi \right)} \right|^2\left| \varPhi _{i}^{\left( D \right)} \right|^2}\right]
\nonumber\\
-\mathrm{SNR}_{\mathrm{in}}&\cdot\sum_{p=1}^P{\sum_{i=1}^I{\varsigma _{i}^{\left( p \right)}\left| \varPhi _{i}^{\left( p \right)} \right|^2\left| \varPhi _{i}^{\left( \varphi \right)} \right|^2\left| \varPhi _{i}^{\left( D \right)} \right|^2}}  
\label{eq.42a}
\\
\mathrm{SDR}_{\mathrm{o}}^{\left( \mathrm{F}-\mathrm{M} \right)}=&\mathrm{SNR}_{\mathrm{in}} \left[MNK - \sum_{i=1}^I{\varrho _i\left| \varPhi _{i}^{\left( \varphi \right)} \right|^4\left| \varPhi _{i}^{\left( D \right)} \right|^2}\right]
\nonumber\\
-\mathrm{SNR}_{\mathrm{in}}&\cdot\sum_{p=1}^P{\sum_{i=1}^I{\varsigma _{i}^{\left( p \right)}\left| \varPhi _{i}^{\left( p \right)} \right|^2\left| \varPhi _{i}^{\left( \varphi \right)} \right|^2\left| \varPhi _{i}^{\left( D \right)} \right|^2}}  
\label{eq.42b}
\end{align}
\end{subequations}
where
\begin{subequations}
\begin{align}
\varPhi _{i}^{\left( \varphi \right)}=&\frac{\sin \left[ N\pi \frac{d}{\lambda}\cos \theta _t\left( \cos \varphi _i-\cos \varphi _t \right) \right]}{\sin \left[ \pi \frac{d}{\lambda}\cos \theta _t\left( \cos \varphi _i-\cos \varphi _t \right) \right]}
\label{eq.43a}
\\
\varPhi _{i}^{\left( D \right)}=&\frac{\sin \left[ K\pi\left( f_{D}^{(i)}-f_D \right) \right]}{\sin \left[\pi \left(f_{D}^{(i)}-f_D \right)\right]} \label{eq.43b}
\\
\varPhi _{i}^{\left( p \right)}=&\frac{\sin \left[ M\pi \frac{2\varDelta f}{c}\left(pR_u-R_t\right) \right]}{\sin \left[ \pi \frac{2\varDelta f}{c}\left(pR_u-R_t\right) \right]}\label{eq.43c}
\end{align}
\end{subequations}
\eqref{eq.43a} and \eqref{eq.43b} suggest that $\varPhi _{i}^{\left( p \right)}$ associated with range is strongly correlated with the performance of range ambiguity clutter suppression. The output SDR is minimum when $\varPhi _{i}^{\left( p \right)}$ is maximum. According to \eqref{eq.43c}, the relationship between the maximum of $\varPhi _{i}^{\left( p \right)}$ and $R_t$ can be expressed as 
\begin{equation}
\mathop {\mathrm{arg}\max} \limits_{R_t}\varPhi _{i}^{\left( p \right)}=p\frac{c}{2T}-\frac{c}{2\varDelta f}L,~L\in \mathbb{N} 
\label{eq.44}
\end{equation}
Assume that the target is located at the range-unambiguous region, $R_t\in \left( 0, R_d \right] $, where $R_d=\frac{c}{2T}$. Then the secondary range ambiguity number $L$ satisfies
\begin{equation}
\left( p-1 \right) \frac{\varDelta f}{T}\leqslant L<p\frac{\varDelta f}{T}
\label{eq.suppp}
\end{equation}
\eqref{eq.suppp} indicates that there are $Z=\lfloor \frac{\varDelta f}{T} \rfloor$ secondary range ambiguity between the $(p-1)$-th ambiguity range and $p$-th ambiguity range. For FDA-MIMO radar, a large $\varDelta f$ determines a large $L$, which means that the performance of range-ambiguous clutter suppression deteriorates by a large number of the secondary range ambiguity. The performance of range-ambiguous clutter suppression significantly depends on the range of CUT. For C-FDA radar, a much smaller frequency offset determines a small number of secondary range ambiguity, which means that the range of CUT has a weak impact on the performance of range-ambiguous clutter suppression. We describe some numerical results in Section VI.


\begin{figure*}[!t]
\centering
\subfloat[PA]{\includegraphics[width=1.78in, height=1.35in]{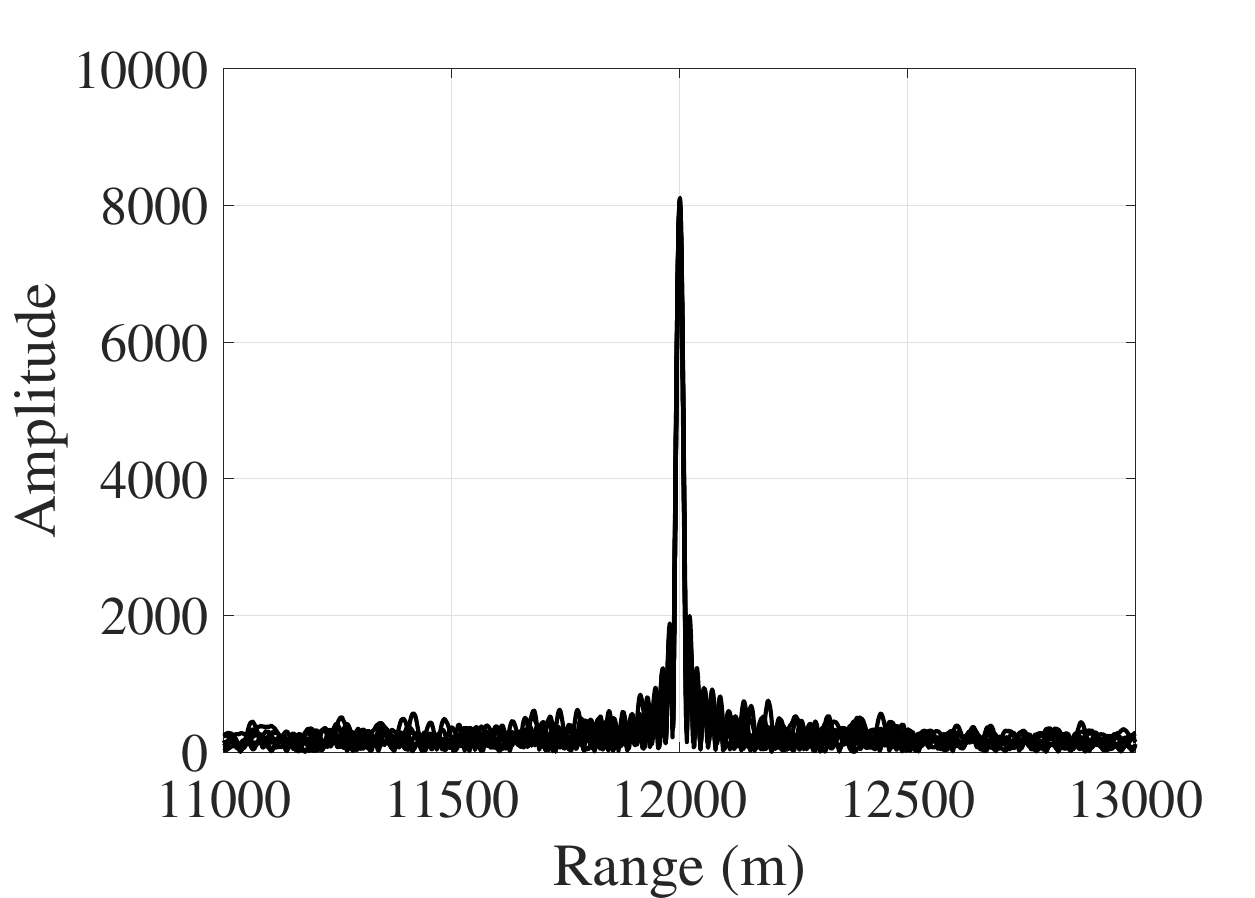}%
\label{fig_3a}}
\hfil
\subfloat[FDA-MIMO]{\includegraphics[width=1.78in, height=1.35in]{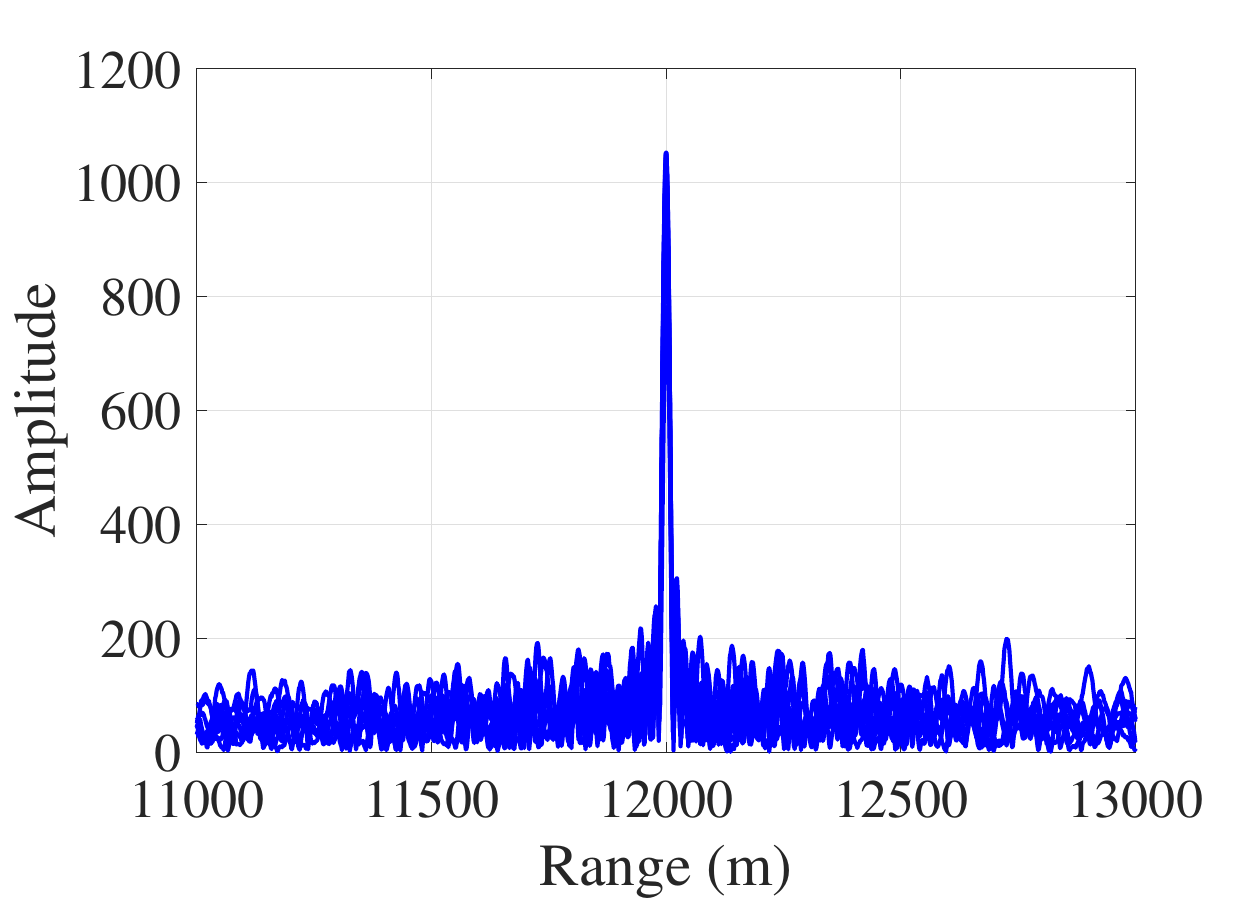}%
\label{fig_3b}}
\hfil
\subfloat[C-FDA ($\varDelta f=100$ kHz)]{\includegraphics[width=1.78in, height=1.35in]{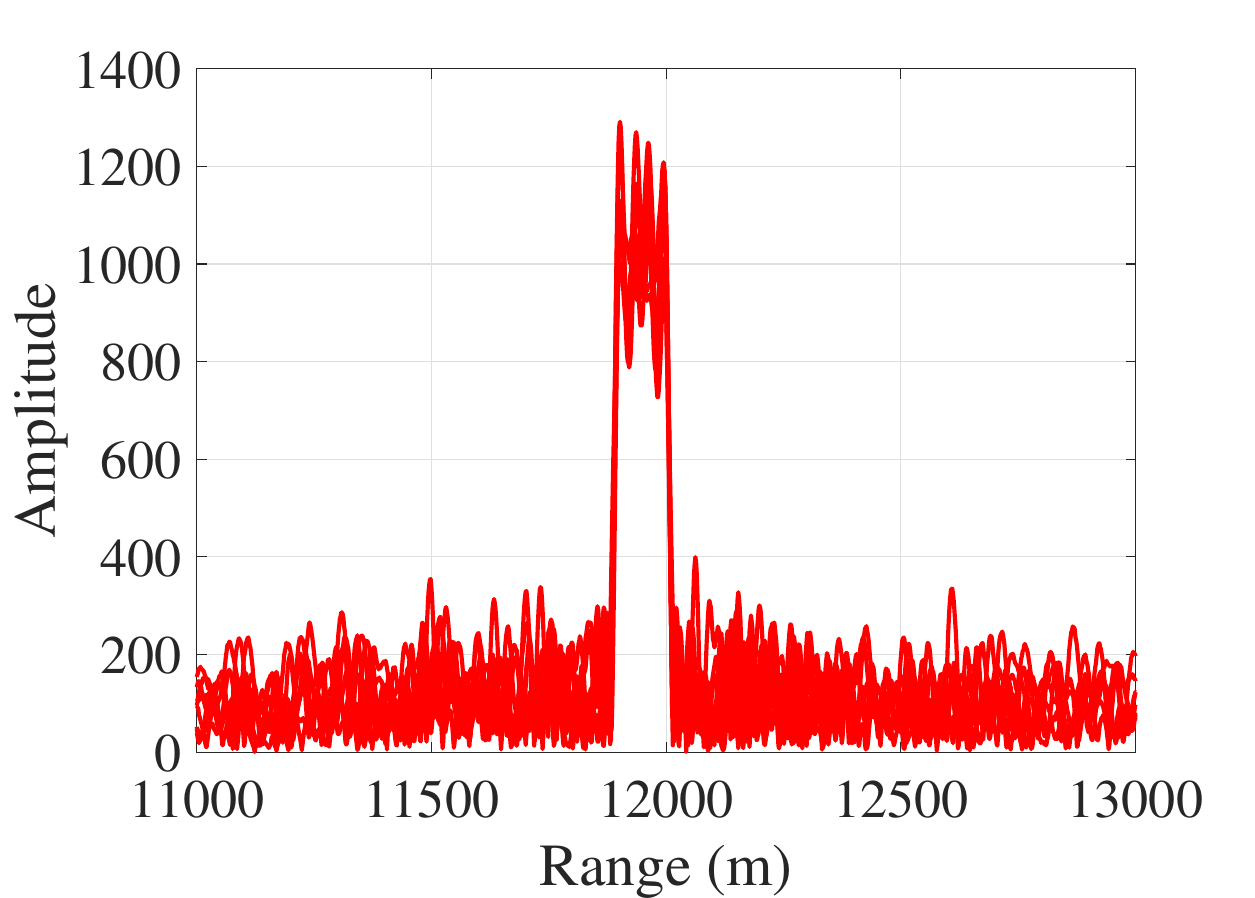}%
\label{fig_3c}}
\hfil
\subfloat[C-FDA ($\varDelta f=50$ kHz)]{\includegraphics[width=1.78in, height=1.35in]{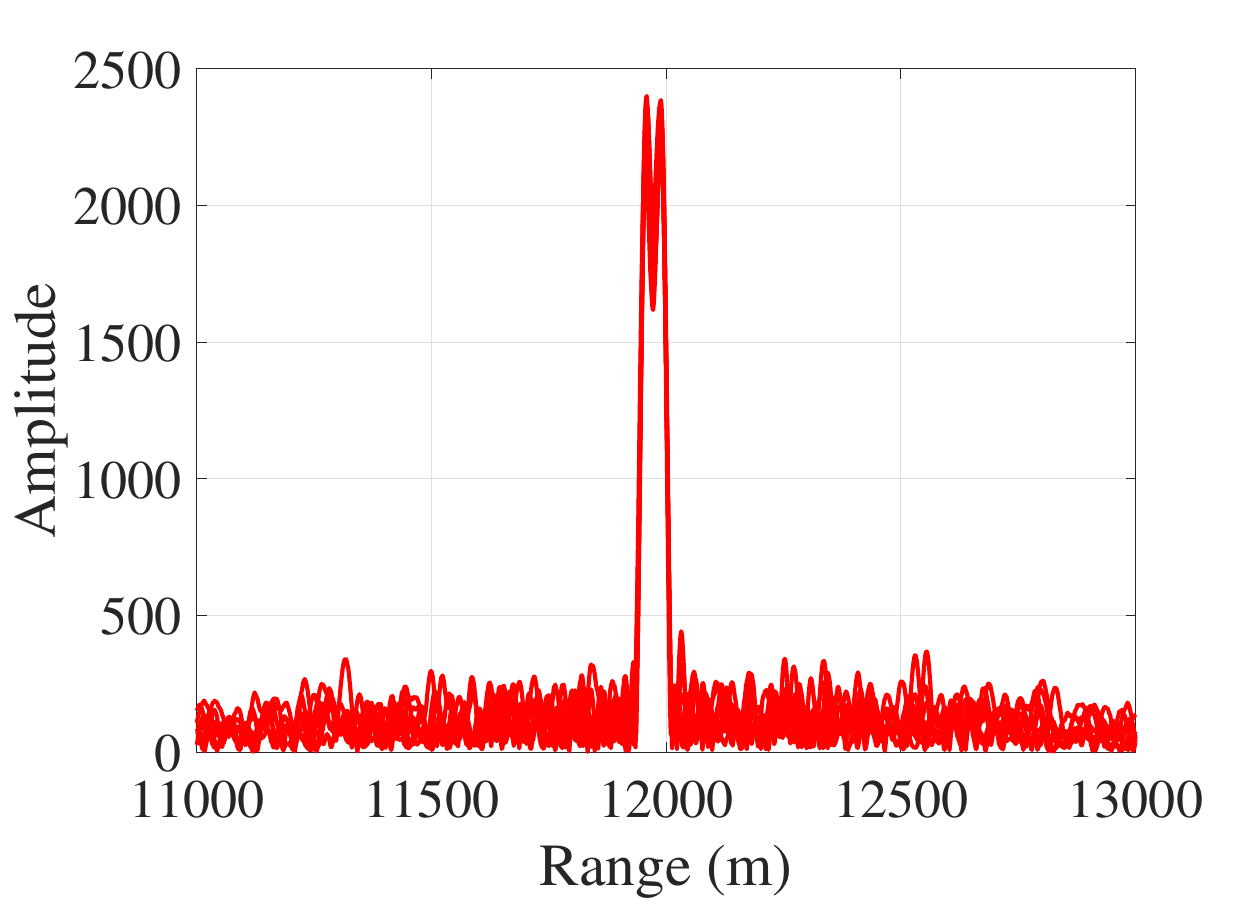}%
\label{fig_3d}}
\hfil
\caption{The range-dimensional output of a single target echo signal of different radar with different radar receiver processing. (a) PA radar receiver processing. (b) FDA-MIMO radar receiver processing. (c) FDA-MIMO radar receiver processing. (d) FDA-MIMO radar receiver processing.}
\label{fig_3}
\end{figure*}

\begin{figure*}[!t]
\centering
\subfloat[C-FDA ($\varDelta f=100$ kHz)]{\includegraphics[width=1.78in, height=1.35in]{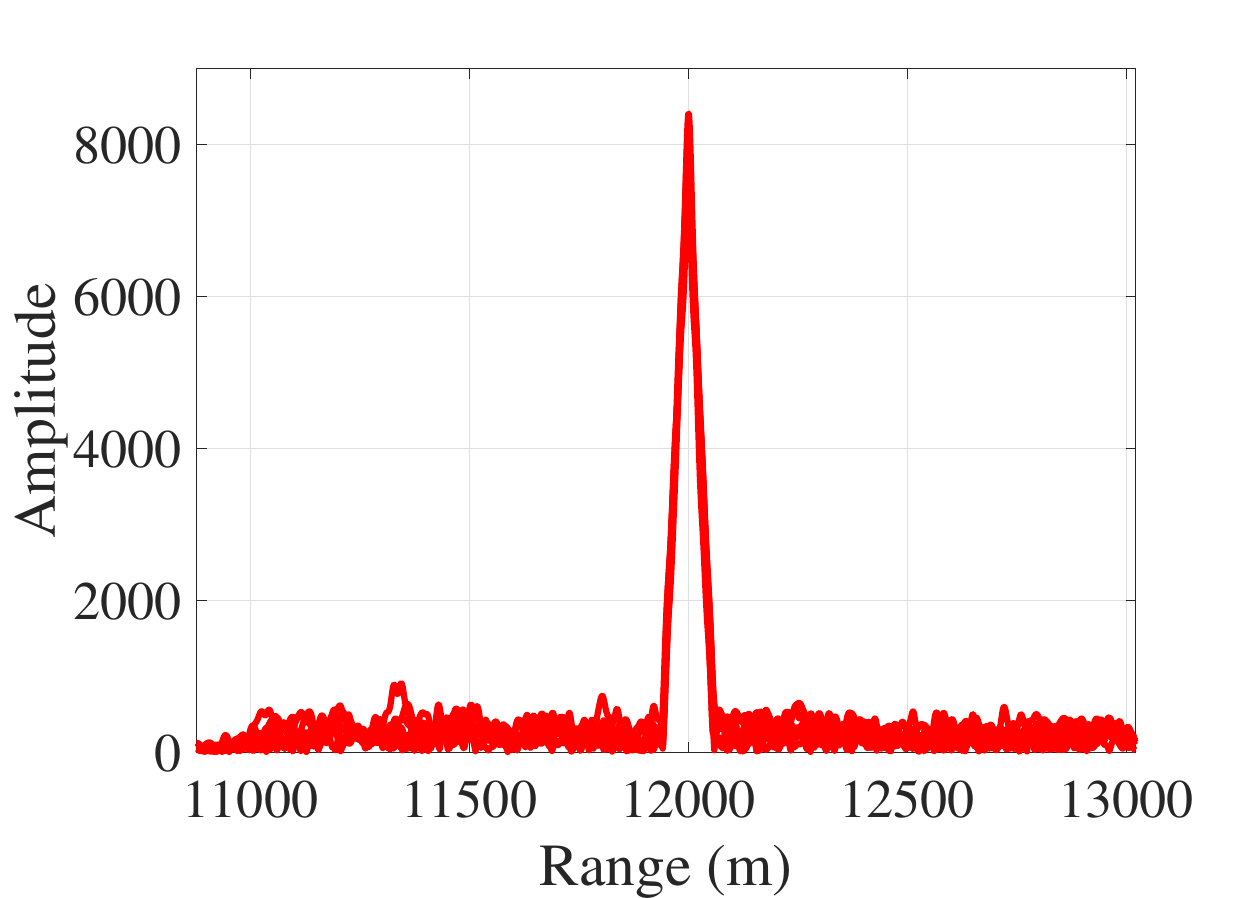}%
\label{fig_4a}}
\hfil
\subfloat[C-FDA ($\varDelta f=50$ kHz)]{\includegraphics[width=1.78in, height=1.35in]{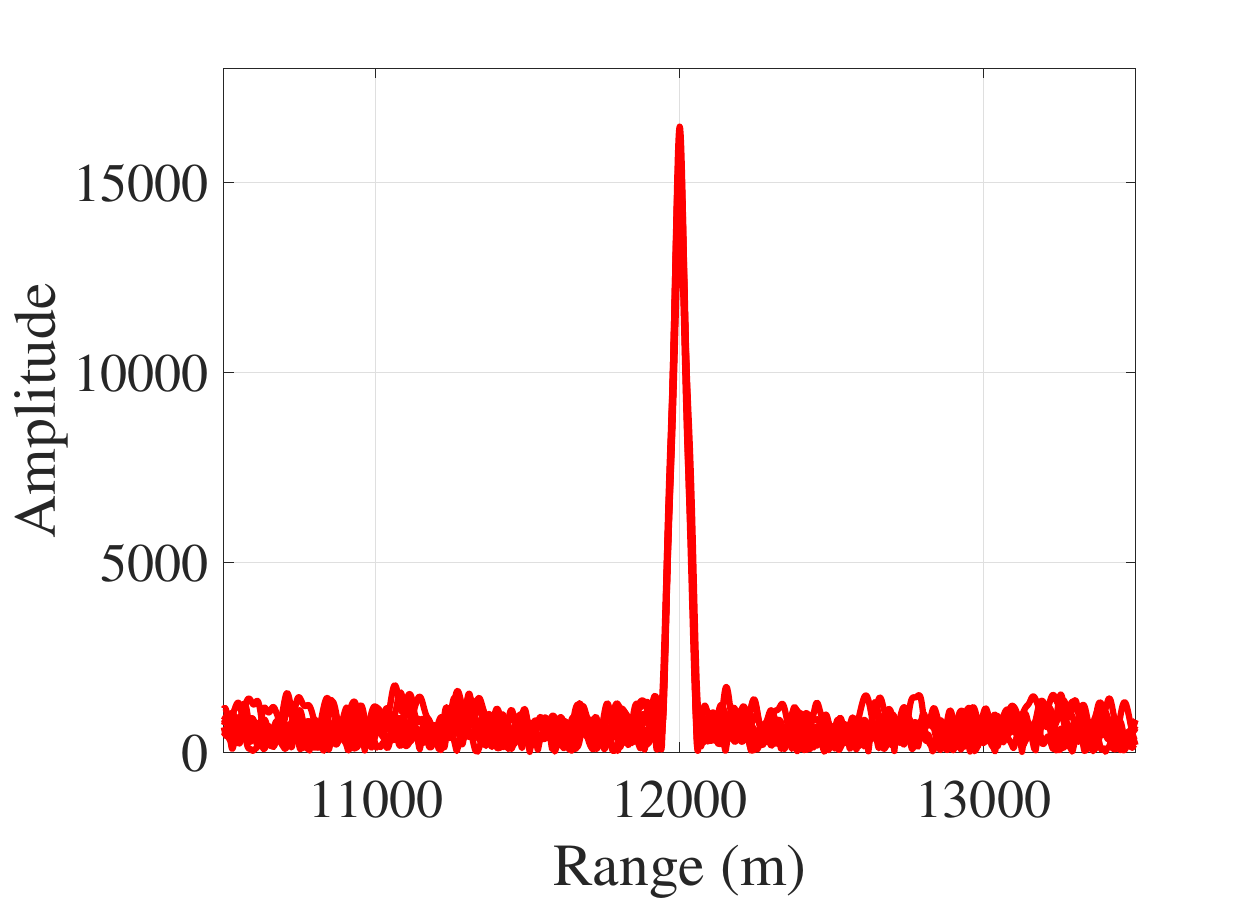}%
\label{fig_4b}}
\hfil
\subfloat[C-FDA ($\varDelta f=10$ kHz)]{\includegraphics[width=1.78in, height=1.35in]{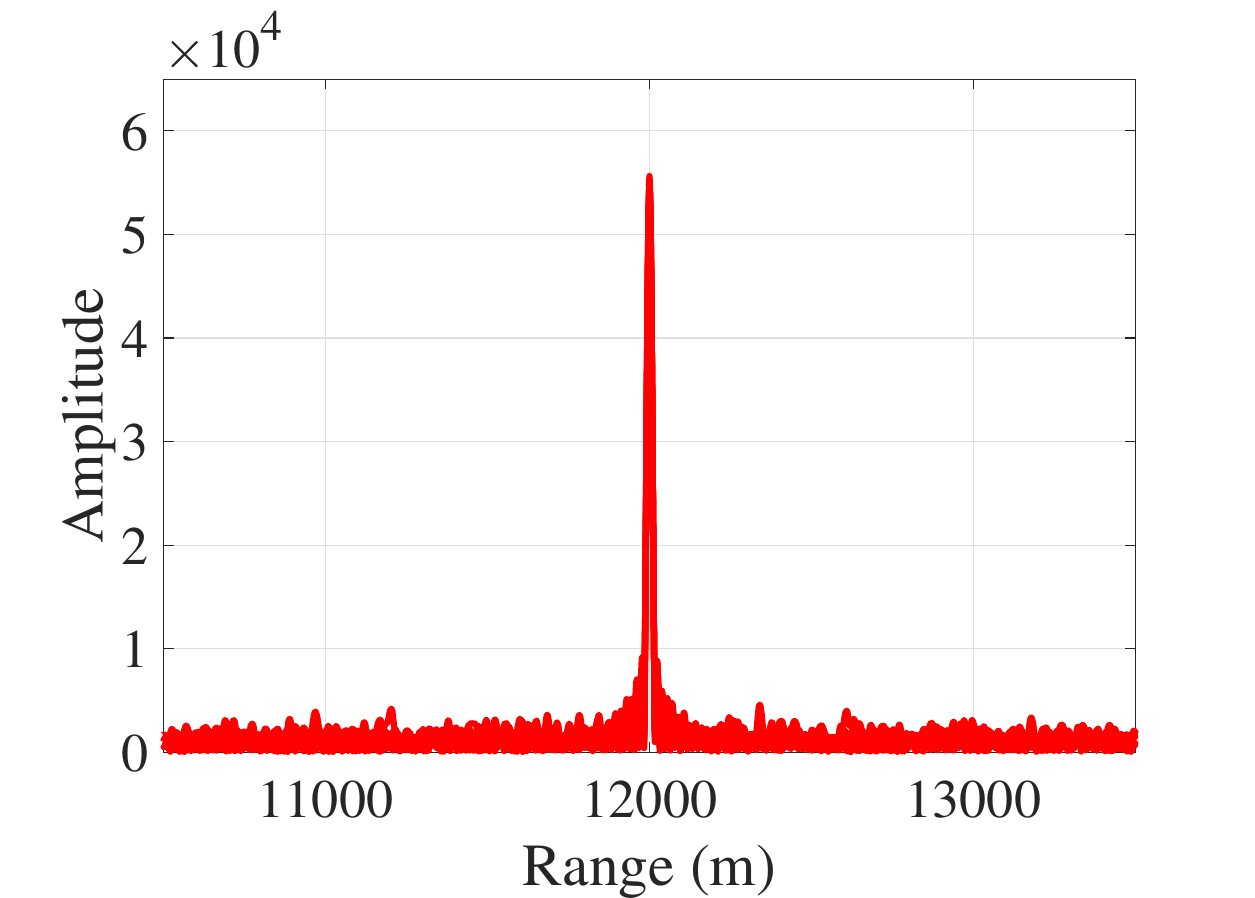}%
\label{fig_4c}}
\hfil
\subfloat[C-FDA ($\varDelta f=0$ kHz)]{\includegraphics[width=1.78in, height=1.35in]{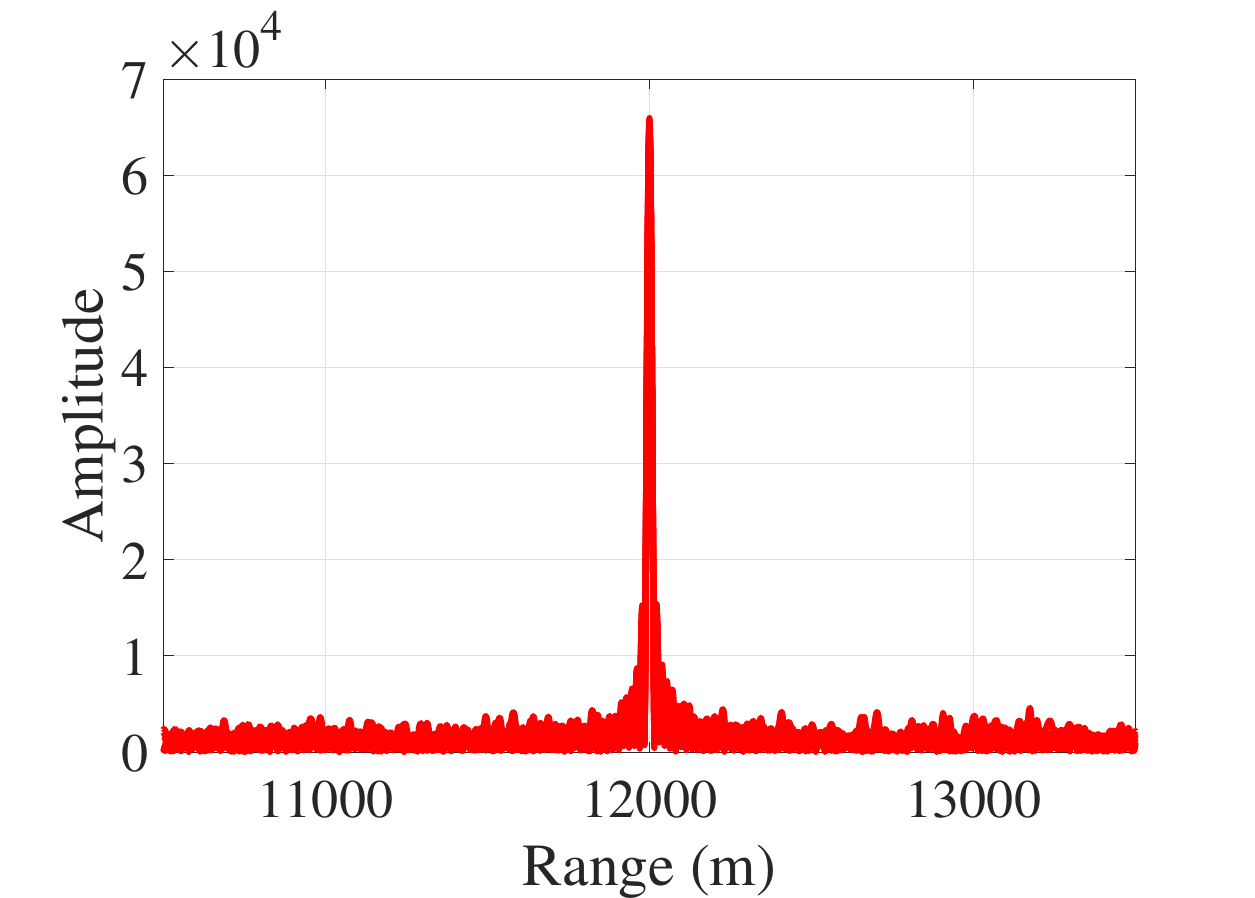}%
\label{fig_4d}}
\caption{The range-dimensional output of a single target echo signal of C-FDA radar with different frequency offset by using the proposed receiver processing. (a) $\varDelta f=100$ kHz. (b) $\varDelta f=50$ kHz. (c) $\varDelta f=10$ kHz. (d) $\varDelta f=0$ kHz.}
\label{fig_4}
\end{figure*}

\section{Numerical results}

\begin{table}[!t]
 \renewcommand{\arraystretch}{1.3}
 \small
 \centering
 \caption{Simulation Parameters}
 \label{tab.1}
 \begin{tabular}{lll}
  \toprule
 Parameter & Symbol & Value  \\
  \midrule
  Carrier frequency & ~~$f_0$ & 10~GHz \\
  Platform height & ~~$H$ & 3000~m\\
  Platform velocity & ~~$v_a$ & 75~m/s\\
  Yaw angle & ~~$\psi$ & 90$^\circ$\\
  Baseband signal bandwidth & ~~$\varDelta f$ & 1~MHz\\
  Radar antenna spacing & ~~$d$ & 15~mm  \\
  Pulse duration & ~~$T_p$ & 10~us\\
  Pulse repetition interval & ~~$T$ & 100~us\\
  Number of transmitting antennas & ~~$M$ & 8\\
  Number of receiving antennas & ~~$N$ & 8\\
  Number of coherent pulses  & ~~$K$ & 8\\
  Azimuth of target & ~~$\varphi_t $ & 0$^\circ$\\
  Range of target & ~~$R_t$ & 12~km\\
  Velocity of target & ~~$v_t$ & 25~m/s\\
  Number of clutter pitches & ~~$I$ & 360\\
  Number of range ambiguity & ~~$P$ & 5\\
  Input signal-to-noise ratio & ~~$\mathrm{SNR}_\mathrm{in}$ & 10~dB\\
  Interference-to-noise ratio & ~~$\mathrm{INR}$ & 30~dB\\
  Clutter-to-noise ratio & ~~$\mathrm{CNR}$ & 50~dB\\
  \bottomrule
 \end{tabular}
\end{table}

In this section, we show the numerical results in terms of receiver processing output, mainlobe interference suppression, and range-ambiguous clutter suppression, to illustrate the performance enhancement of the proposed C-FDA radar compared with the conventional PA radar, MIMO radar, and FDA-MIMO radar. The simulation parameters and corresponding symbols of radar, target, jamming, and clutter is presented in Table \ref{tab.1}. In addition, we use the LFM signal with frequency modulation ratio $\kappa=B/T_p=10^{10}$ as the radar baseband waveform signal and set a frequency offset of $\varDelta f=1$ MHz for FDA-MIMO radar to satisfy waveform orthogonality. Note that the radar range resolution is $c/(2B)=150$ m for the bandwidth of $1$ MHz and the ambiguity range is $R_u=cT/2=15$ km. For FDA-MIMO radar, the secondary ambiguity range is $c/(2\varDelta f)=150$ m.

\subsection{Output of receiver processing}

In this subsection, we present the range-dimensional output of a single target echo signal from a noisy environment by using different radar receiver processing. The PA radar receiver processing is described in \cite{ref1} and\cite{ref4} and the multiple-channel receiver processing of FDA-MIMO radar is described in \cite{ref17},\cite{ref36} and\cite{ref37}. Note that all simulations illustrate the range-dimensional output in the first channel of the first receive element for FDA-MIMO radar and C-FDA radar, and the results of the remaining channels and receive elements are similar. Note that the output amplitude of one transmit signal is related to the Fast Fourier Transform (FFT) sample rate $f_s=100$ MHz and pulse width $T_p=10$ us, namely $f_s\cdot T_p=1000$.

In Fig.\ref{fig_3}, we show the range-dimensional output of a mono-target echo signal of different radars with different radar receiver processing. Firstly, the output result of PA radar with its receiver processing is presented in Fig.\ref{fig_3a}. The target peak can be obtained at $12$ km with an amplitude of 8000, which is consistent with the array gain of the 8-array PA radar. Secondly, the output result of FDA-MIMO radar with its own receiver processing is shown in Fig.\ref{fig_3b}. The target peak can be obtained at $12$ km with an amplitude of 1000, which is consistent with the analysis of array gain between FDA-MIMO and PA radar in \eqref{eq.12}. Thirdly, we illustrate the range-dimensional output of C-FDA radar of $\varDelta f=100$ kHz and $\varDelta f=50$ kHz with FDA-MIMO radar receiver processing in Fig.\ref{fig_3c} and Fig.\ref{fig_3d}, respectively. In Fig.\ref{fig_3c}, multiple peaks appear at different ranges associated with the target signal and the peak amplitudes are roughly 1000 to 1200. In Fig.\ref{fig_3d}, multiple peaks are closer to $12$ km and their amplitudes increase to around 2200 as the C-FDA radar frequency offset decreases to $50$ kHz. However, these two output results are unsatisfactory for further target detection and parameter estimation. 

\begin{figure*}[!t]
\centering
\subfloat[FDA-MIMO before suppression]{\includegraphics[width=1.78in, height=1.35in]{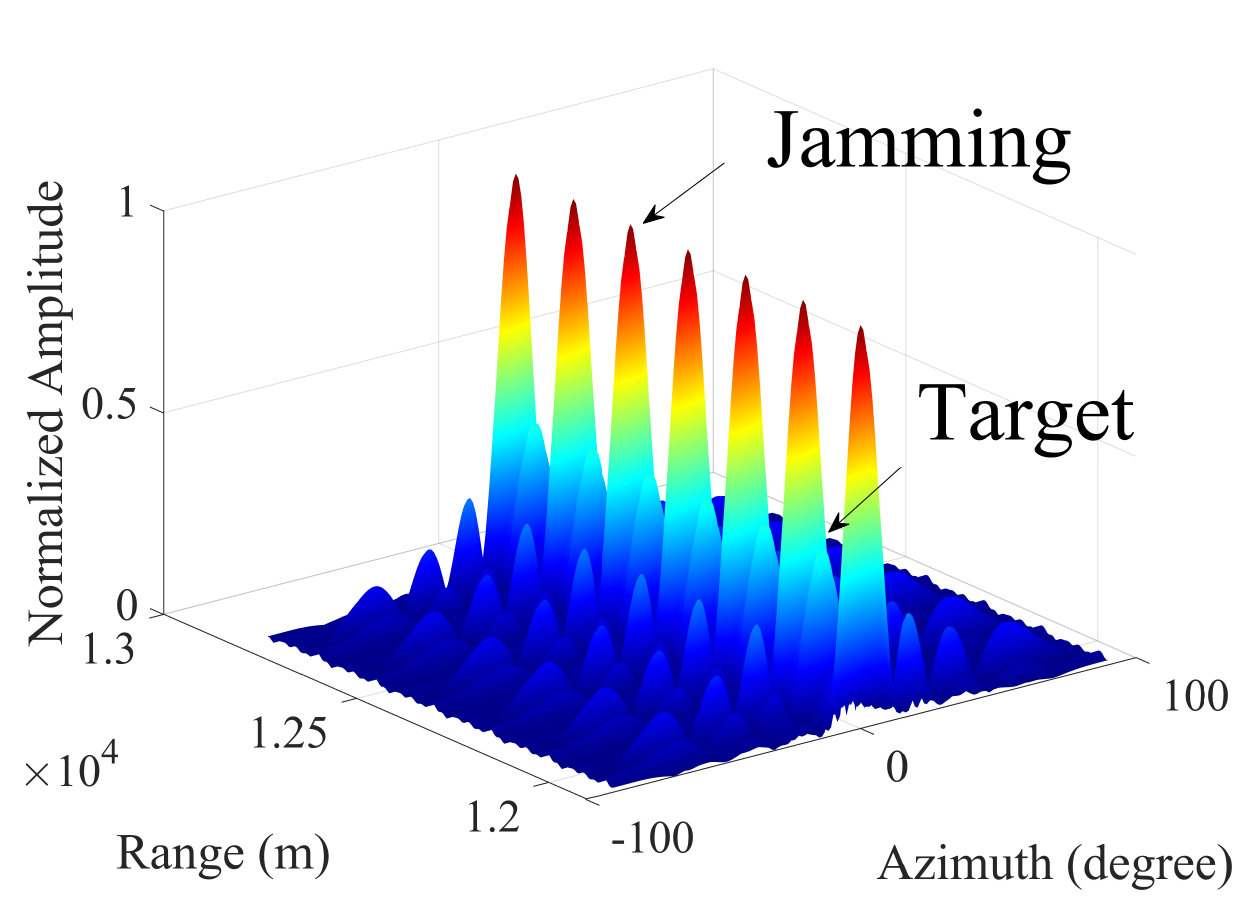}%
\label{fig_5a}}
\hfil
\subfloat[FDA-MIMO after suppression]{\includegraphics[width=1.78in, height=1.35in]{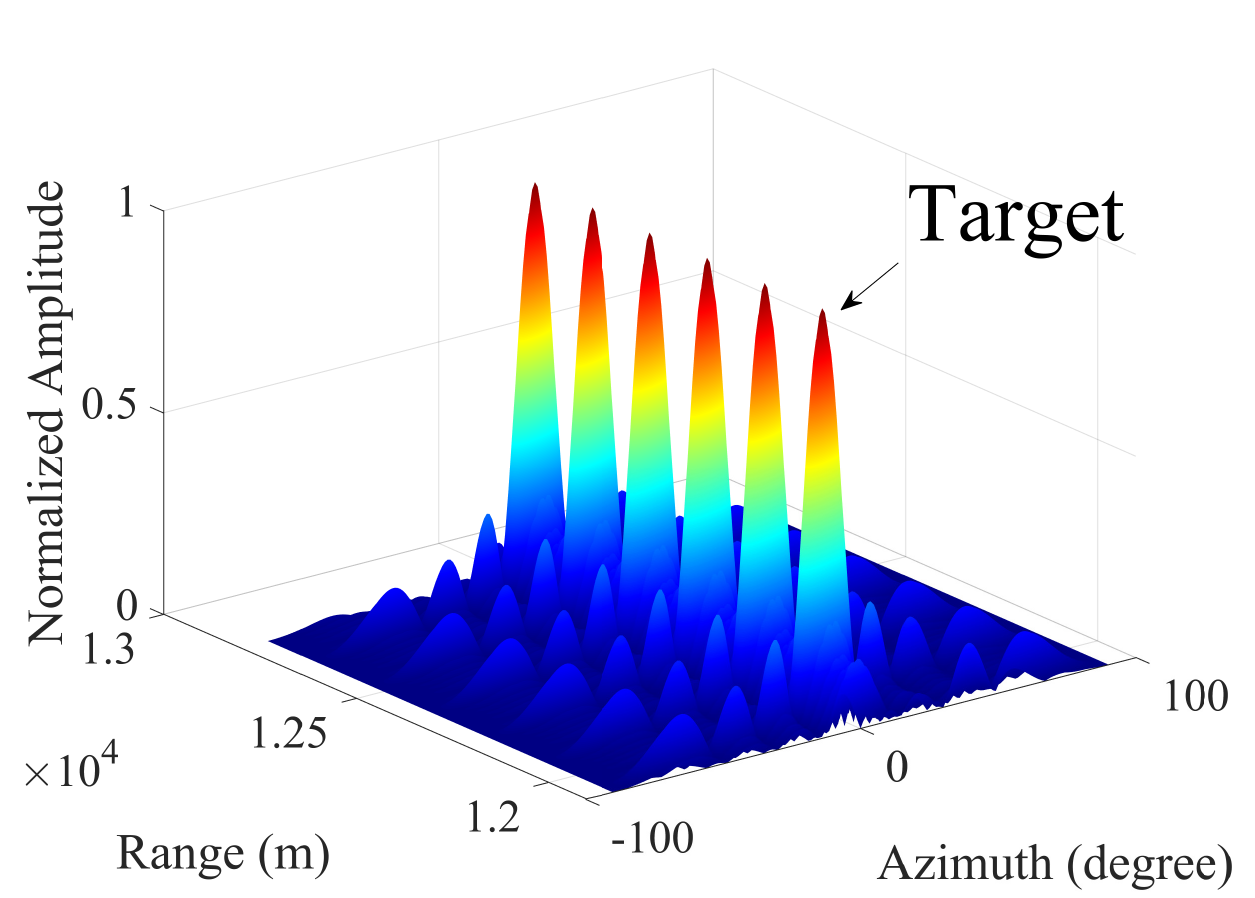}%
\label{fig_5b}}
\hfil
\subfloat[C-FDA before suppression]{\includegraphics[width=1.78in, height=1.35in]{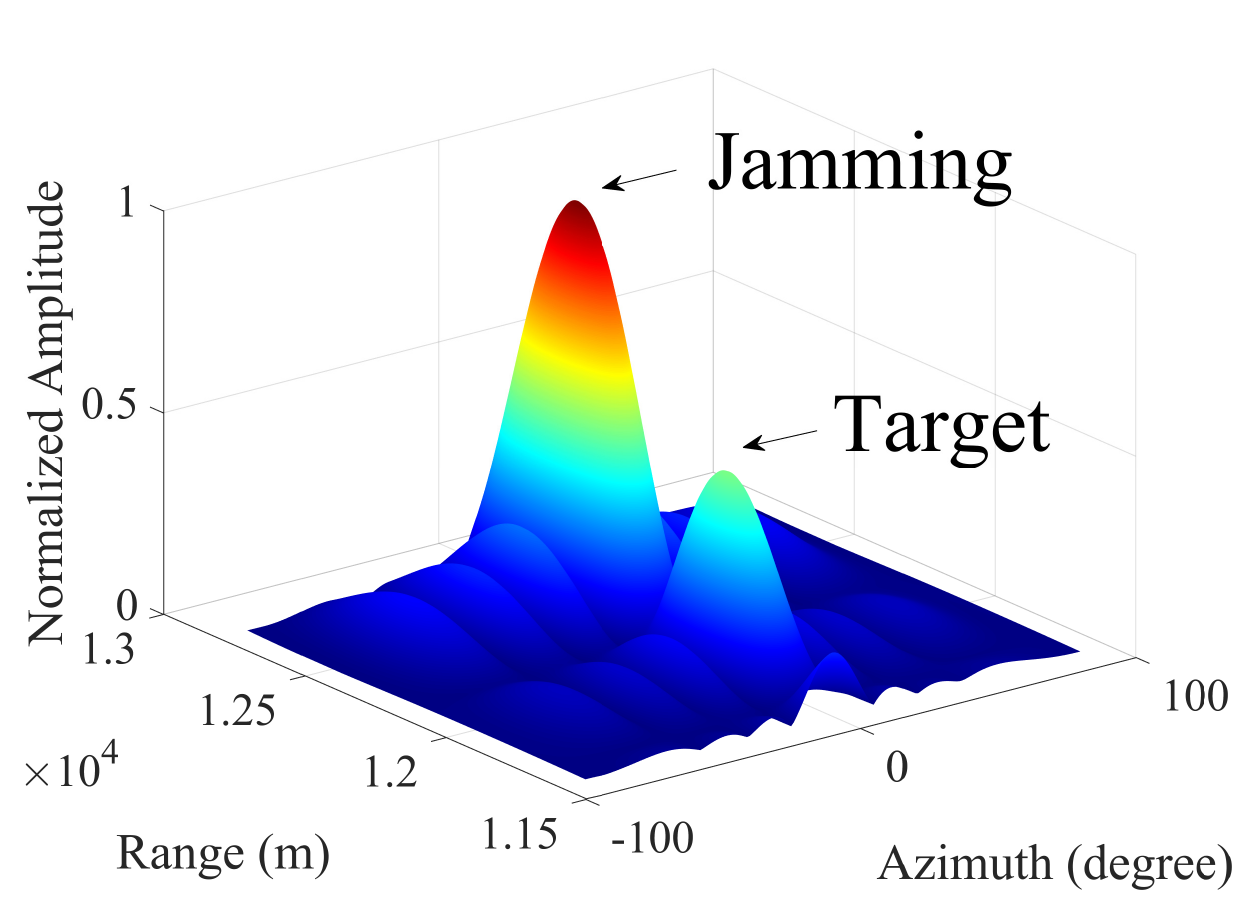}%
\label{fig_5c}}
\hfil
\subfloat[C-FDA after suppression]{\includegraphics[width=1.78in, height=1.35in]{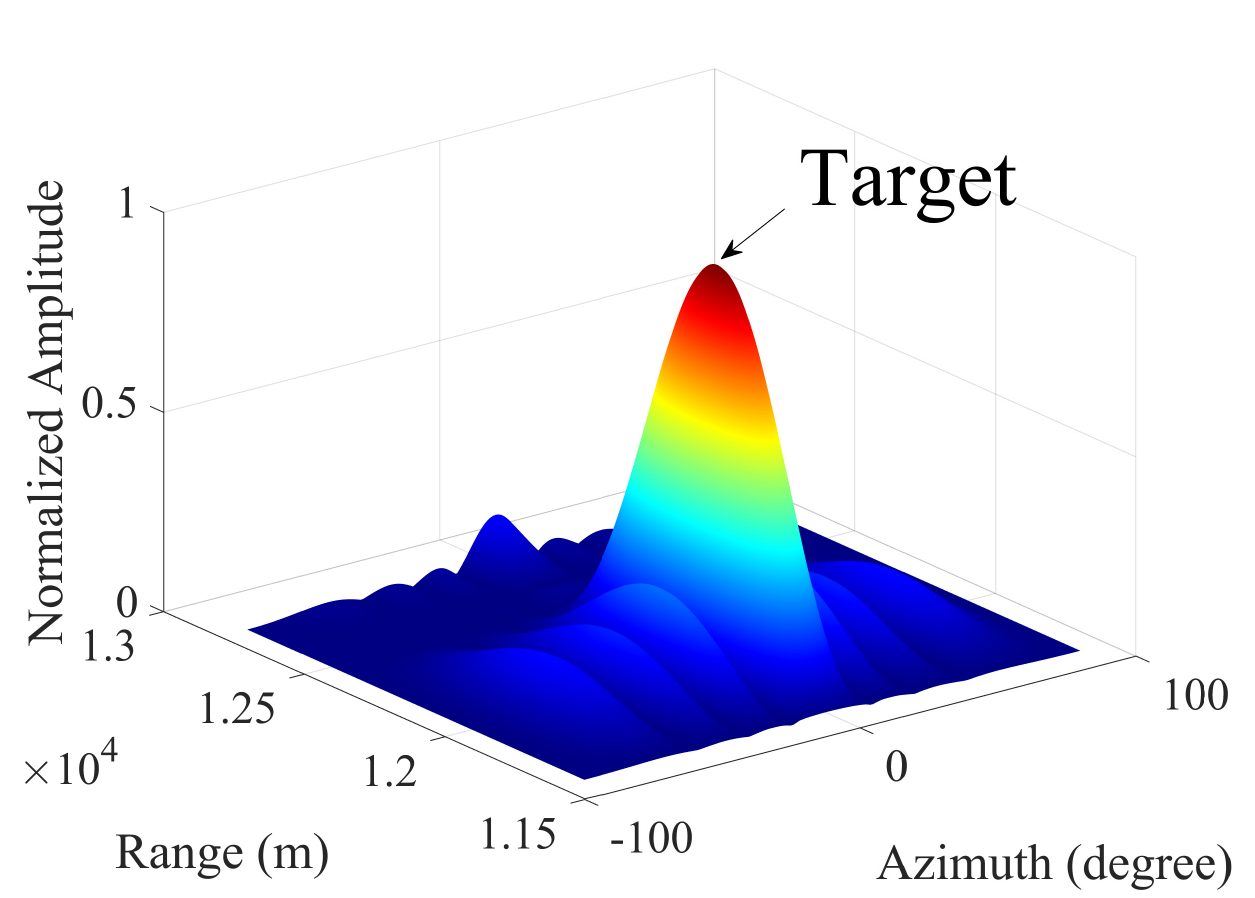}%
\label{fig_5d}}
\caption{The range-azimuth three-dimensional spectrum before and after interference suppression when jamming is at 12.5 km. (a) before suppression for FDA-MIMO. (b) after suppression for FDA-MIMO. (c) before suppression for C-FDA ($50$ kHz). (d) after suppression for C-FDA ($50$ kHz).}
\label{fig_5}
\end{figure*}

\begin{figure*}[!t]
\centering
\subfloat[FDA-MIMO before suppression]{\includegraphics[width=1.78in, height=1.35in]{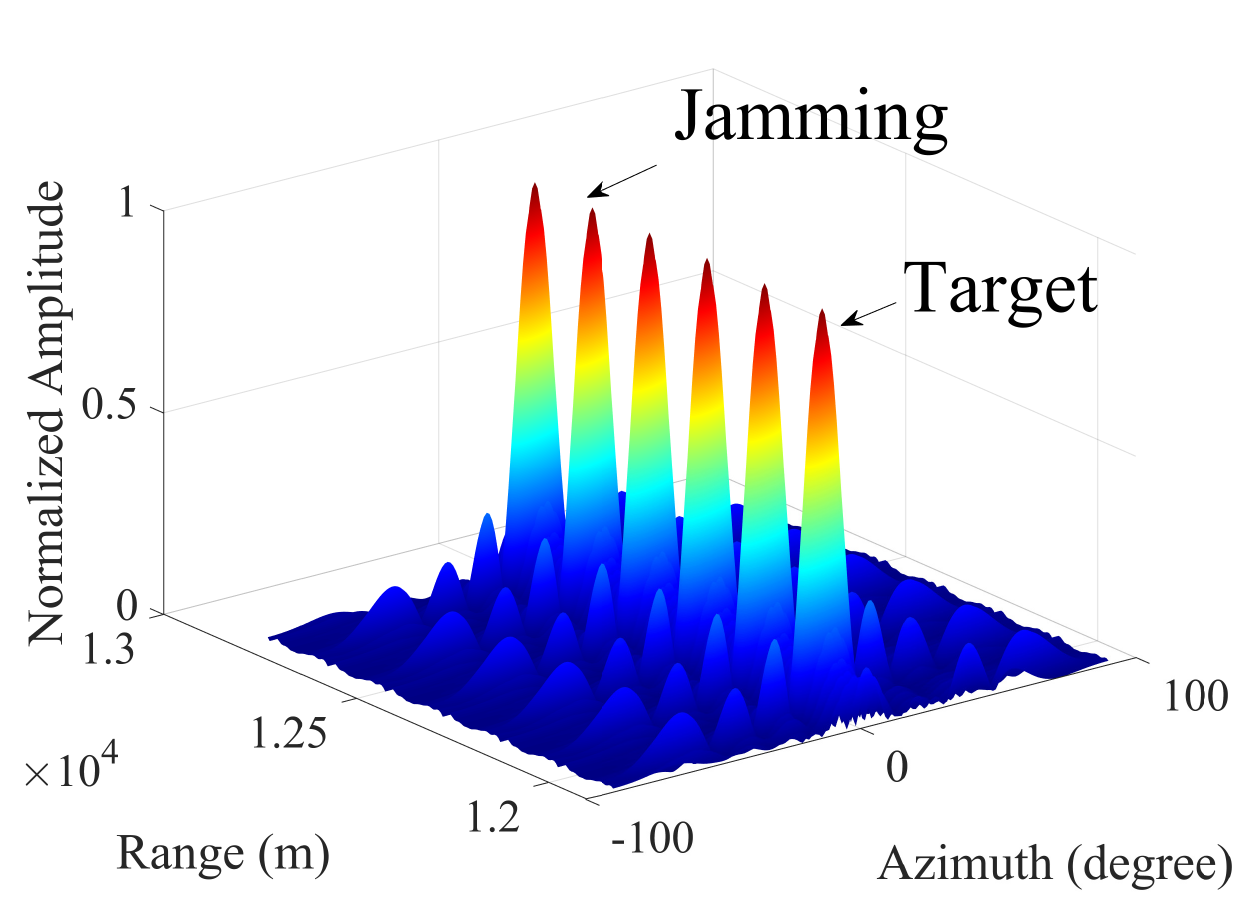}%
\label{fig_6a}}
\hfil
\subfloat[FDA-MIMO after suppression]{\includegraphics[width=1.78in, height=1.35in]{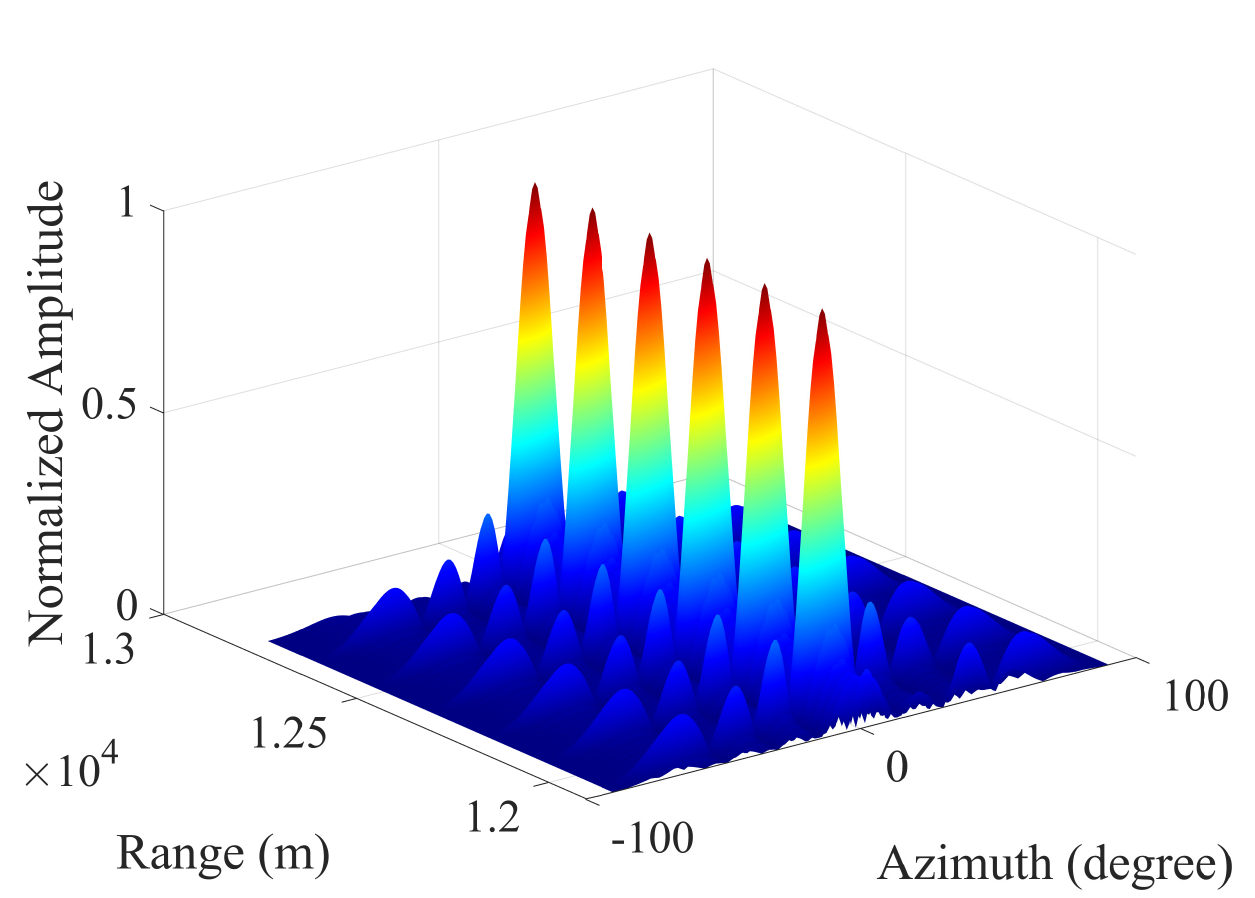}%
\label{fig_6b}}
\hfil
\subfloat[C-FDA before suppression]{\includegraphics[width=1.78in, height=1.35in]{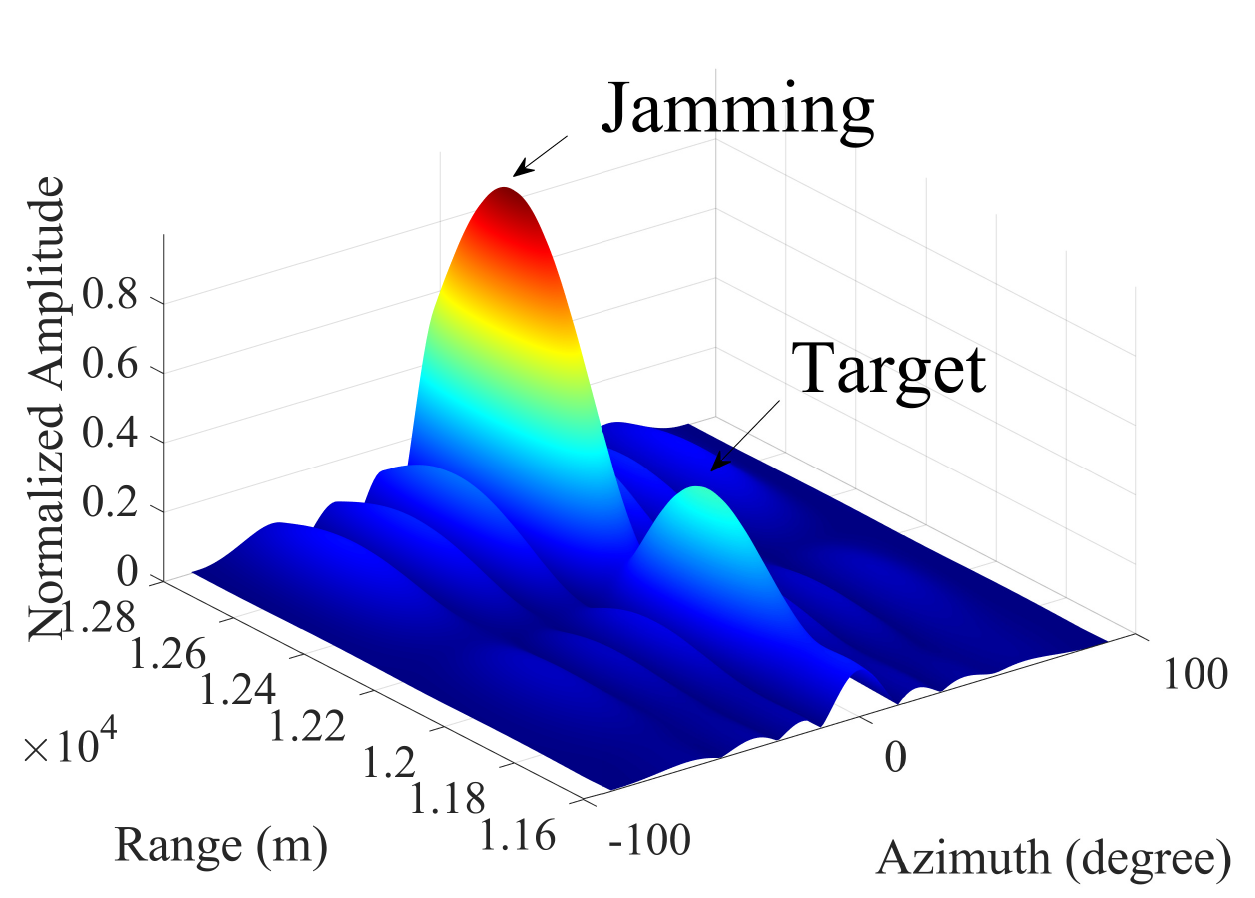}%
\label{fig_6c}}
\hfil
\subfloat[C-FDA after suppression]{\includegraphics[width=1.78in, height=1.35in]{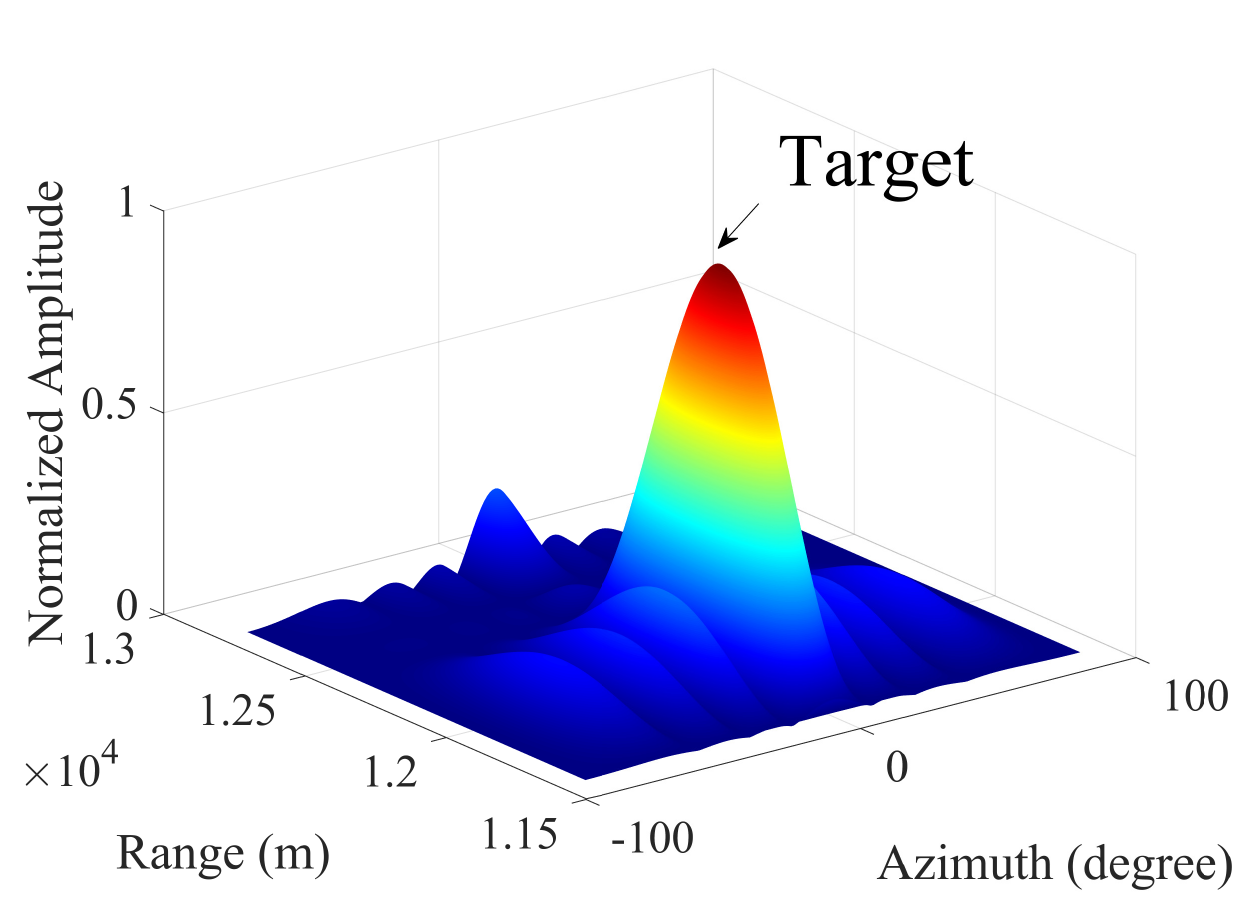}%
\label{fig_6d}}
\hfil
\caption{The range-azimuth three-dimensional spectrum before and after interference suppression when jamming is at 12.6 km. (a) before suppression for FDA-MIMO. (b) after suppression for FDA-MIMO. (c) before suppression for C-FDA ($50$ kHz). (d) after suppression for C-FDA ($50$ kHz).}
\label{fig_6}
\end{figure*}

In Fig.\ref{fig_4}, we show the range-dimensional output of the same target echo in Fig.\ref{fig_3} for C-FDA radar of different frequency offset with the proposed receiver processing. In Fig.\ref{fig_4a}, we show the output result of C-FDA radar with a frequency offset of $\varDelta f=$100 kHz, appearing one target peak at 12 km with amplitude over 8000. In Fig.\ref{fig_4b}, we show the output result of C-FDA radar with a frequency offset of $\varDelta f=$50 kHz, appearing one target peak at 12 km with amplitude over 15000. Compared with Fig.\ref{fig_3c} and Fig.\ref{fig_3d}, the proposed C-FDA receiver processing effectively solves the problem of multiple peaks and low amplitudes. In Fig.\ref{fig_4c} and Fig.\ref{fig_4d}, we present the output result of C-FDA radar with a frequency offset of $\varDelta f=$10 kHz and $\varDelta f=$0 kHz, where the peak amplitudes at 12 km exceed $5.5\times 10^4$ and $6\times 10^4$, respectively. Moreover, the output amplitudes of C-FDA radar with different frequency offset in Fig.\ref{fig_4} are larger than the amplitude of PA radar in Fig.\ref{fig_3a} and the output amplitude of the C-FDA radar with frequency offset of $\varDelta f=$0 kHz in Fig.\ref{fig_4d} is near $6\times 10^4$, which is consistent with the comparison of array gains in \eqref{eq.22} of Theorem \ref{Thm.2}.

\subsection{Mainlobe interference suppression}

In this section, we show the performance comparison of mainlobe interference suppression. Target and jamming are designed to have the same azimuth of $0^{\circ}$ and velocity of $75$ m/s but different ranges. We use the spectrum to describe the distribution of target and jamming energy in a range-azimuth two-dimensional plane before and after mainlobe interference suppression.
\begin{equation}
P_F=\frac{1}{\boldsymbol{s}^H\boldsymbol{Q}^{-1}\boldsymbol{s}}
\label{eq.45}
\end{equation}
where $\boldsymbol{Q}$ is the target plus jamming covariance matrix and $\boldsymbol{s}$ is a copy of target steering vector $\boldsymbol{t}$ for matching in the range-azimuth three-dimensional space \cite{ref4}. In addition, we use the aforementioned output SINR to evaluate the performance of mainlobe interference suppression for the proposed C-FDA radar and those conventional radars.

In Fig.\ref{fig_5}, we show the range-azimuth three-dimensional spectrum before and after mainlobe interference suppression when the target is at $12$ km but the jamming is at $12.5$ km, comparing the FDA-MIMO and C-FDA radar. For FDA-MIMO radar, Fig.\ref{fig_5a} illustrates the target and jamming appear periodically in the range dimension with the secondary ambiguity range $c/(2\varDelta f)=150$ m as a period. In Fig.\ref{fig_5b}, the jamming energy is suppressed due to $\varDelta R=\left| R_j-R_t \right|\ne {c}L/{2\varDelta f}, L\in \mathbb{N}$, which is discussed in \eqref{eq.28}. For C-FDA radar with a frequency offset of $50$ kHz, Fig.\ref{fig_5c} illustrates that the target and jamming are explicitly presented at 12 km and 12.5 km without secondary range-ambiguous and Fig.\ref{fig_5d} shows the jamming energy is suppressed and the target remains at 12 km.

\begin{figure}[t]
\centering
\includegraphics[width=3.7in]{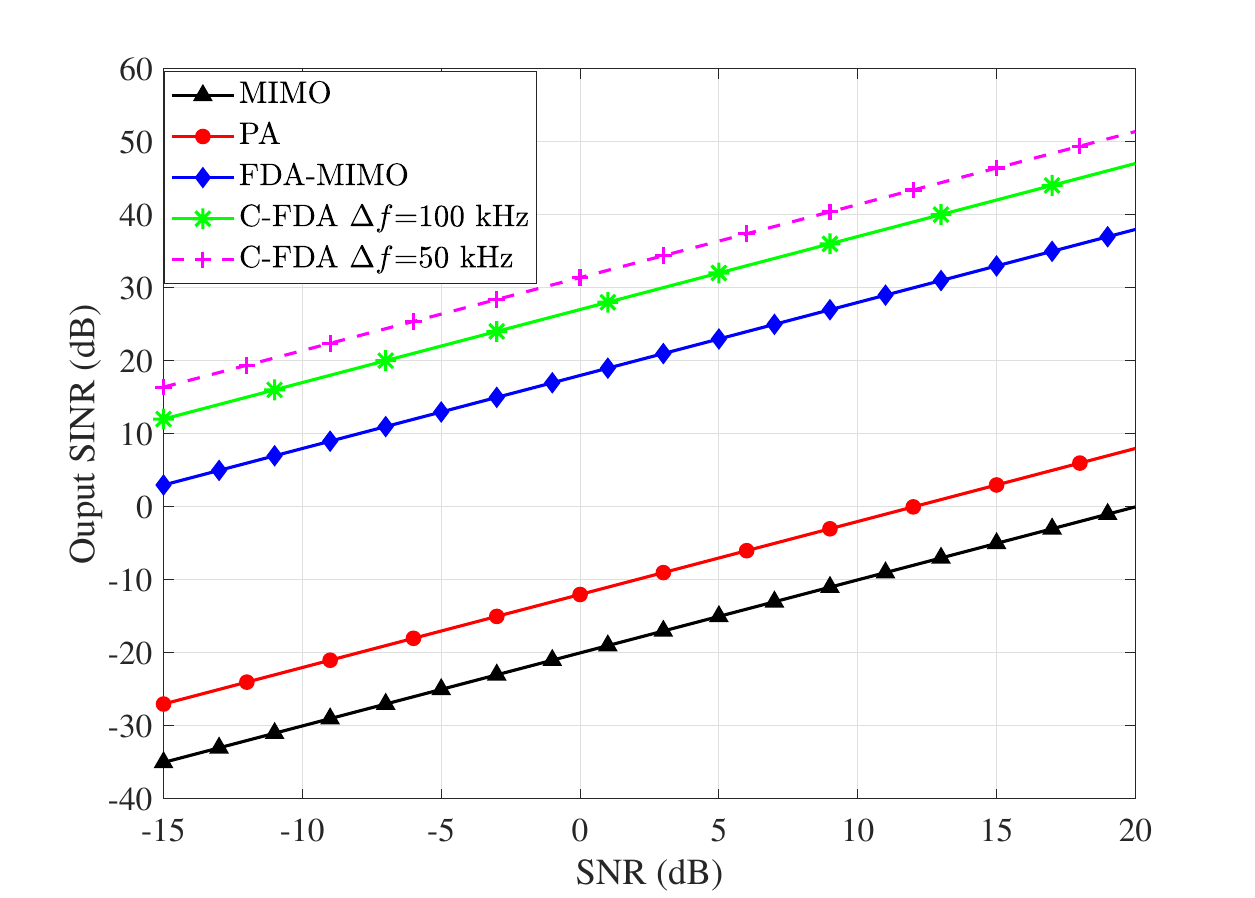}
\caption{Output SINR versus SNR when the jamming is at 12.5 km.}
\label{fig_7}
\end{figure}

As a comparison, we present the range-azimuth spectrum before and after interference suppression when the jamming is at $12.6$ km in Fig\ref{fig_6}, where $\varDelta R=\left| R_j-R_t \right|= 4\times{c}/{(2\varDelta f)}$. In Fig.\ref{fig_6a}, jamming is overlapped by the periodic repetition of target energy due to the secondary range-ambiguous of FDA-MIMO radar. Thereby the jamming energy can not be suppressed in the Fig.\ref{fig_6b}. For C-FDA radar with a frequency offset of $50$ kHz, Fig.\ref{fig_6c} illustrates that the target and jamming are explicitly presented at 12 km and 12.6 km without secondary range-ambiguous and Fig.\ref{fig_6d} shows the jamming energy is suppressed and the target remains at 12 km. 

Fig.\ref{fig_7} shows the output SINR versus input SNR when the jamming is at 12.5 km, which is consistent with Fig.\ref{fig_5}, comparing the proposed C-FDA radar with conventional radars, namely PA radar, MIMO radar, and FDA-MIMO radar. Note that the output SINRs of PA and MIMO radar are significantly lower because of their range-independency and the output SINR of the PA radar is larger than that of the MIMO radar due to its coherent array gain. For the case where the jamming is at 12.5 km, both FDA-MIMO and C-FDA radar have a high output SINR after mainlobe interference suppression, while the output SINR of the proposed C-FDA radar is higher with the decrease of frequency offset due to its higher coherent array gain as discussed in Theorem.\ref{Thm.2}.

Fig.\ref{fig_8} illustrates the output SINR versus SNR when the jamming is at 12.6 km, which is consistent with Fig.\ref{fig_6}. Focusing on FDA-MIMO radar, the output SINR is reduced to the same level as PA and MIMO radar, indicating that the mainlobe interference cannot be suppressed when the jamming locates at the secondary ambiguity range of the target. Fig.\ref{fig_7} and Fig.\ref{fig_8} suggest that the proposed C-FDA can effectively suppress the mainlobe interference located at different ranges, without deteriorating from the secondary range-ambiguous and with a high coherent processing gain.

\begin{figure}[t]
\centering
\includegraphics[width=3.7in]{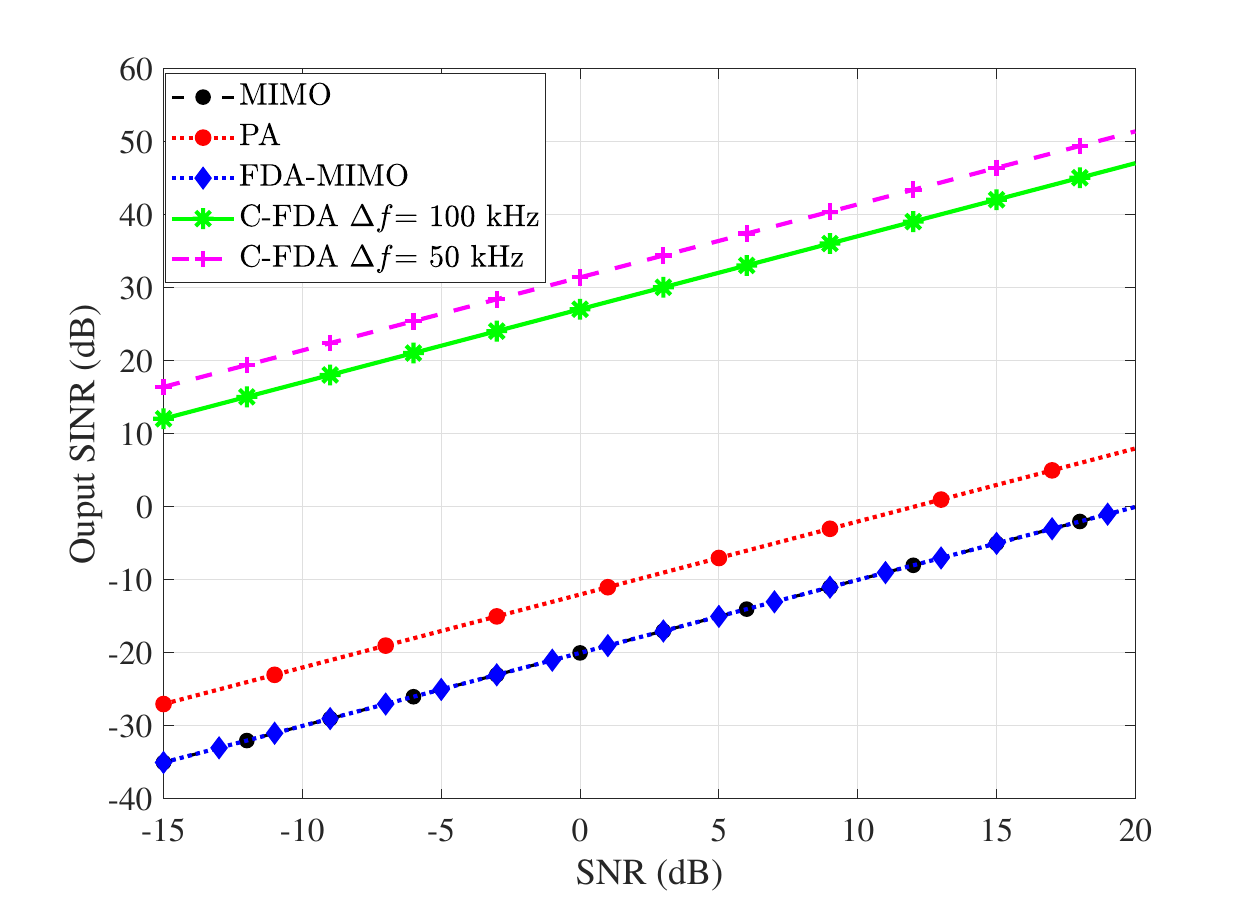}
\caption{Output SINR versus SNR when the jamming is at 12.6 km.}
\label{fig_8}
\end{figure}

\subsection{Range-ambiguous clutter suppression}

In this subsection, we present the range-ambiguous clutter suppression results to show the performance of the proposed C-FDA radar with respect to the FDA-MIMO radar. Note that the target is located at the range-unambiguous region $R_t=12~\mathrm{km}<15~\mathrm{km}$ and the clutter-to-noise ratio (CNR) of CUT is defined as
\begin{equation}
\mathrm{CNR}\triangleq \frac{\mathrm{tr}\left( \boldsymbol{R}_d \right)}{\sigma _{\mathrm{n}}^{2}\cdot MNK}
\label{eq.46}
\end{equation}
where $\boldsymbol{R}_d$ is referred in \eqref{eq.36}. Furthermore, we use SDR loss to evaluate the performance of range-ambiguous clutter suppression, which is calculated by the output SDR to the input SDR measured in the CUT.
\begin{equation}
\mathrm{SDR}_{\mathrm{loss}}=\frac{\mathrm{SDR}_{\mathrm{o}}}{\mathrm{SDR}_{\mathrm{i}}}=\mathrm{SDR}_{\mathrm{o}}\frac{\mathrm{SNR}_{\mathrm{in}}}{1+\mathrm{CNR}}
\label{eq.47}
\end{equation}
where $\mathrm{SDR}_{\mathrm{o}}$ can be replaced by \eqref{eq.42a} for C-FDA radar and \eqref{eq.42b} for FDA-MIMO radar. According to the radar resolution range of $150$ m, the target is located at the 800th range bin.

\begin{figure*}[!t]
\centering
\subfloat[FDA-MIMO radar]{\includegraphics[width=1.78in, height=1.35in]{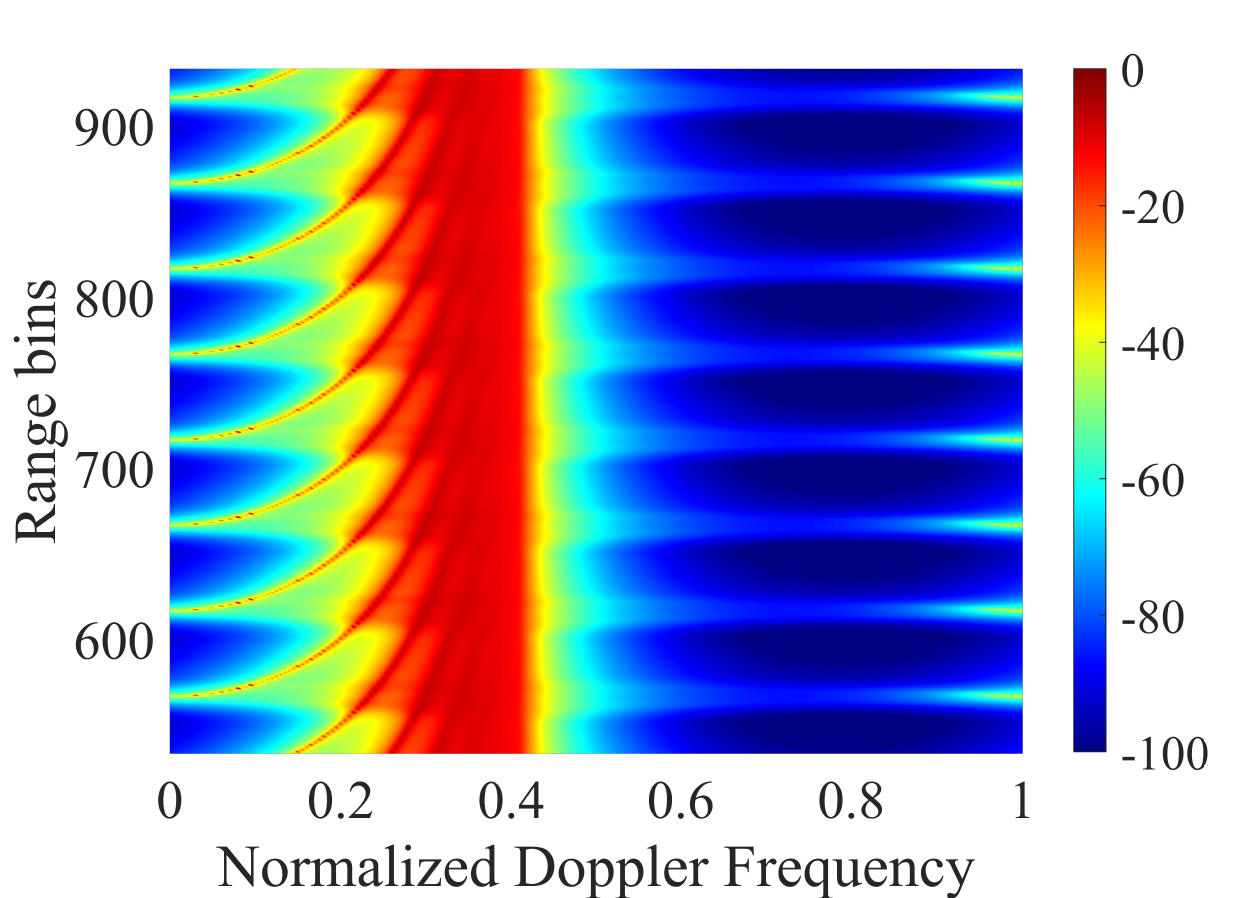}%
\label{fig_9a}}
\hfil
\subfloat[C-FDA ($200$ kHz)]{\includegraphics[width=1.78in, height=1.35in]{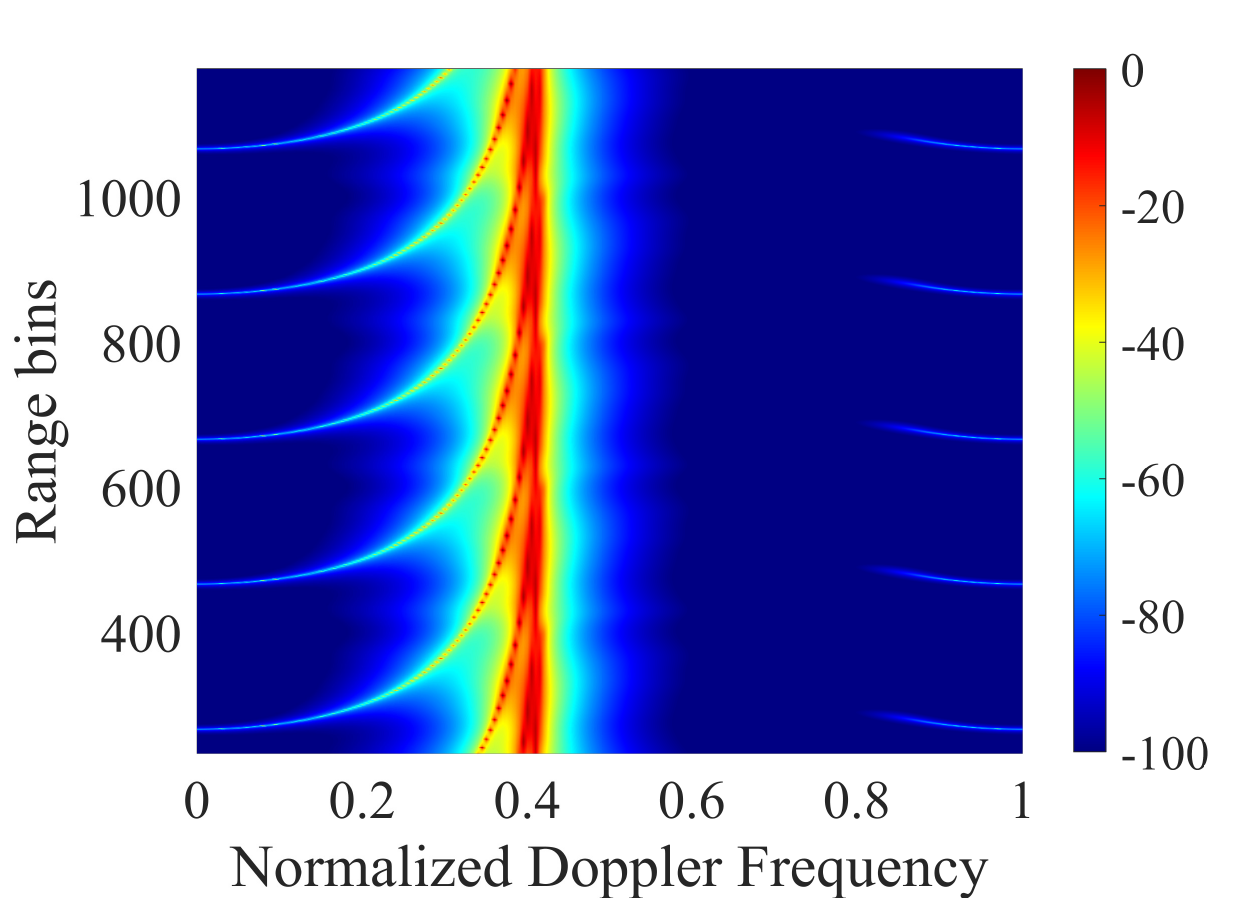}%
\label{fig_9b}}
\hfil
\subfloat[C-FDA ($100$ kHz)]{\includegraphics[width=1.78in, height=1.35in]{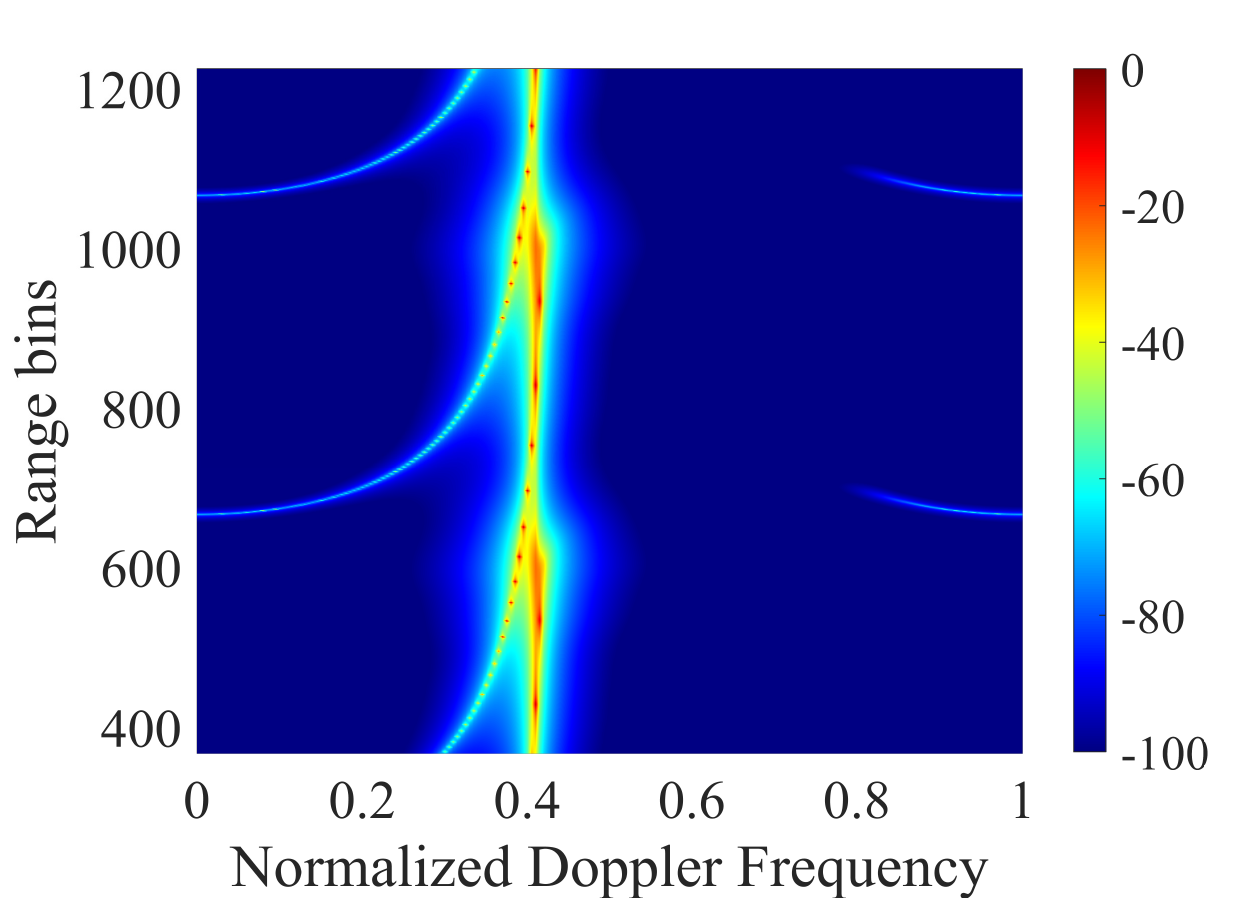}%
\label{fig_9c}}
\hfil
\subfloat[C-FDA ($50$ kHz)]{\includegraphics[width=1.78in, height=1.35in]{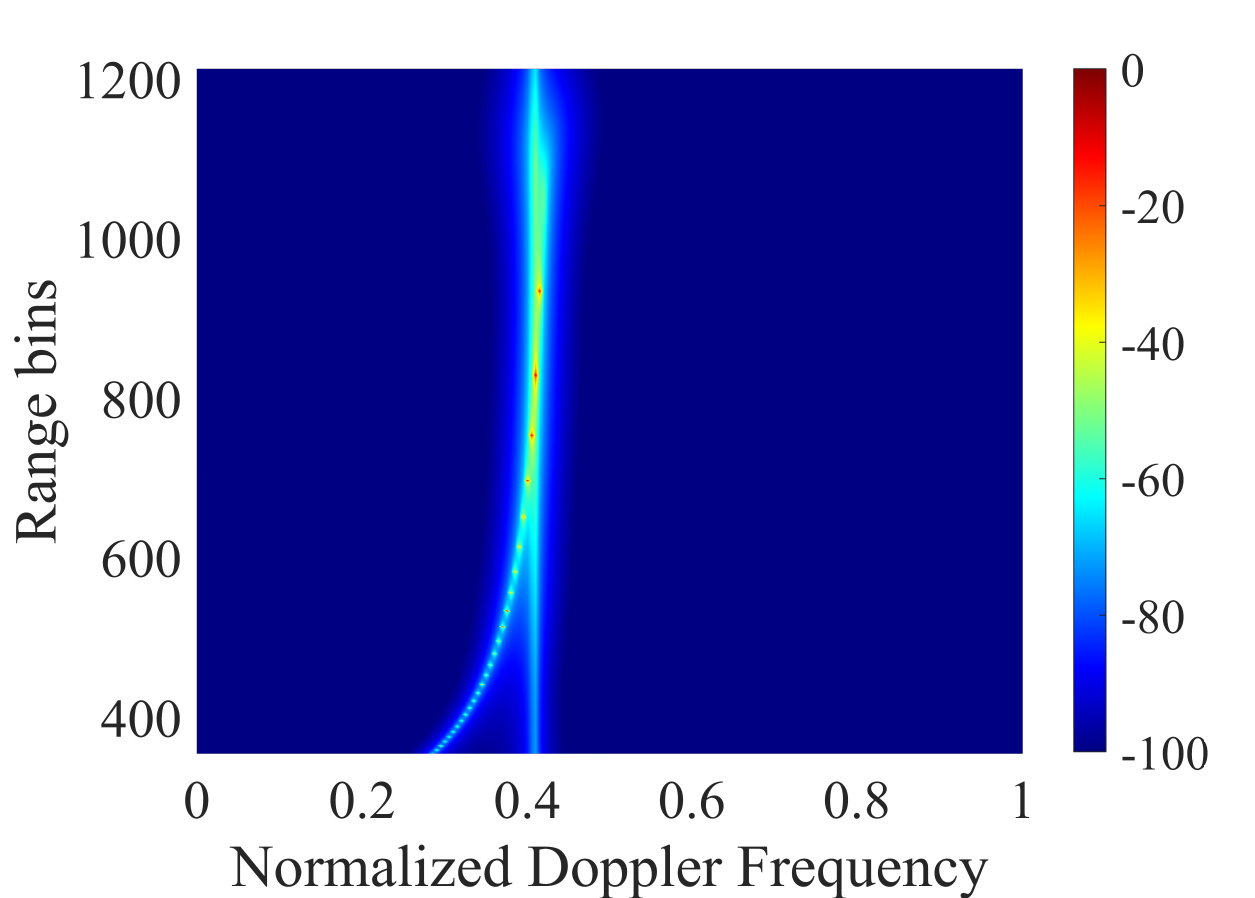}%
\label{fig_9d}}
\hfil
\caption{The range-Doppler two-dimensional clutter spectrum for FDA-MIMO radar and C-FDA radar with different frequency offset. (a) FDA-MIMO radar. (b) C-FDA ($200$ kHz). (c) C-FDA ($100$ kHz). (d) C-FDA ($50$ kHz).}
\label{fig_9}
\end{figure*}

\begin{figure*}[!t]
\centering
\subfloat[FDA-MIMO]{\includegraphics[width=1.78in, height=1.35in]{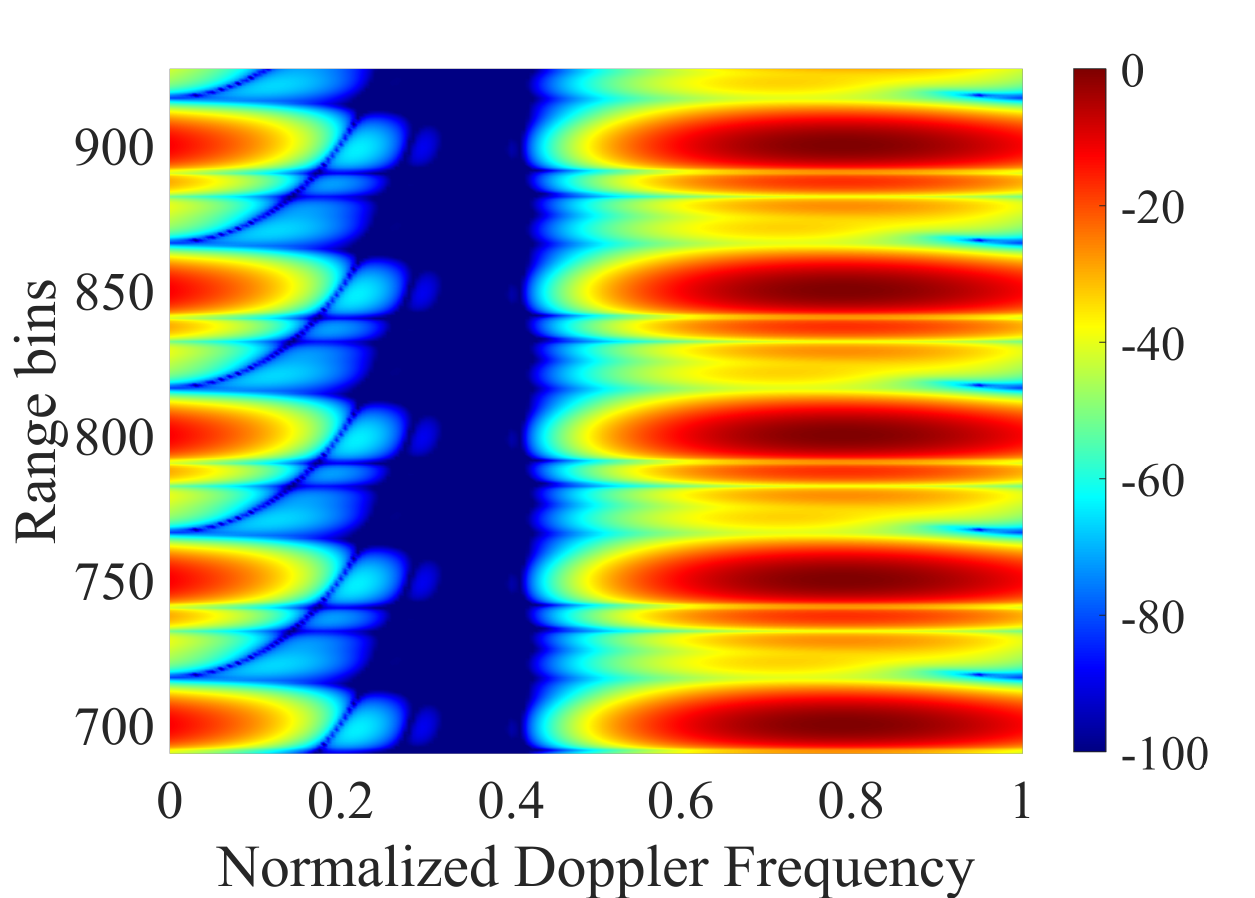}%
\label{fig_10a}}
\hfil
\subfloat[C-FDA ($200$ kHz)]{\includegraphics[width=1.78in, height=1.35in]{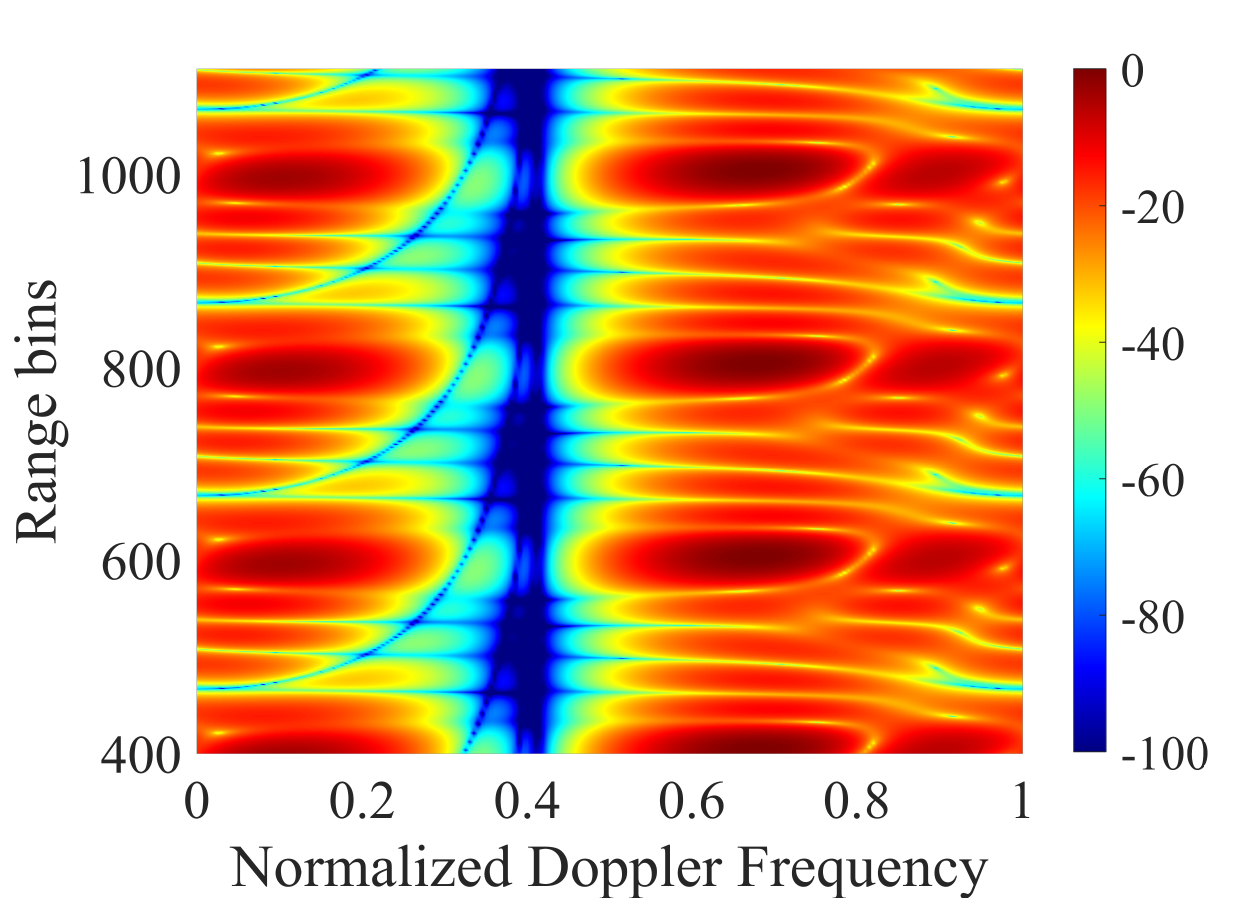}%
\label{fig_10b}}
\hfil
\subfloat[C-FDA ($100$ kHz)]{\includegraphics[width=1.78in, height=1.35in]{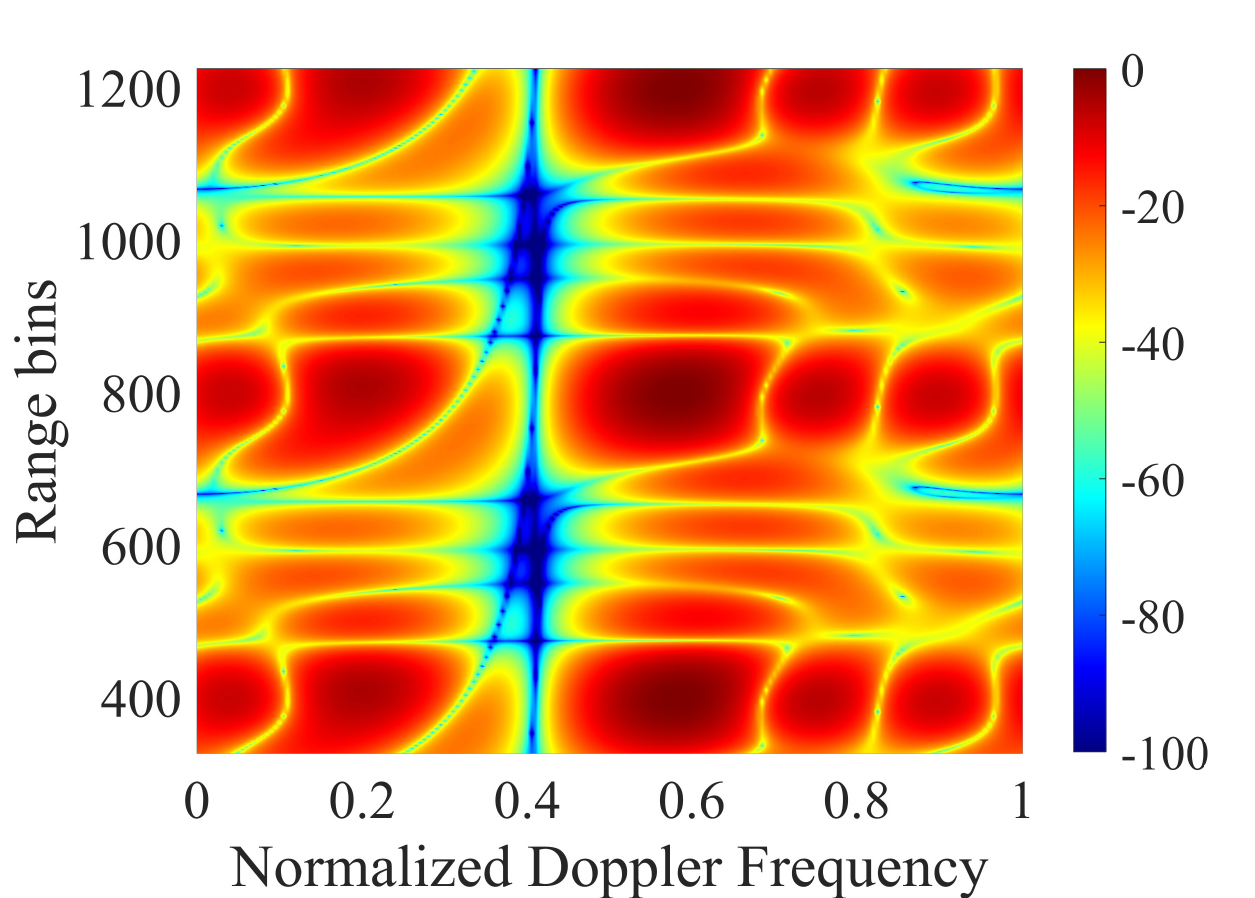}%
\label{fig_10c}}
\hfil
\subfloat[C-FDA ($50$ kHz)]{\includegraphics[width=1.78in, height=1.35in]{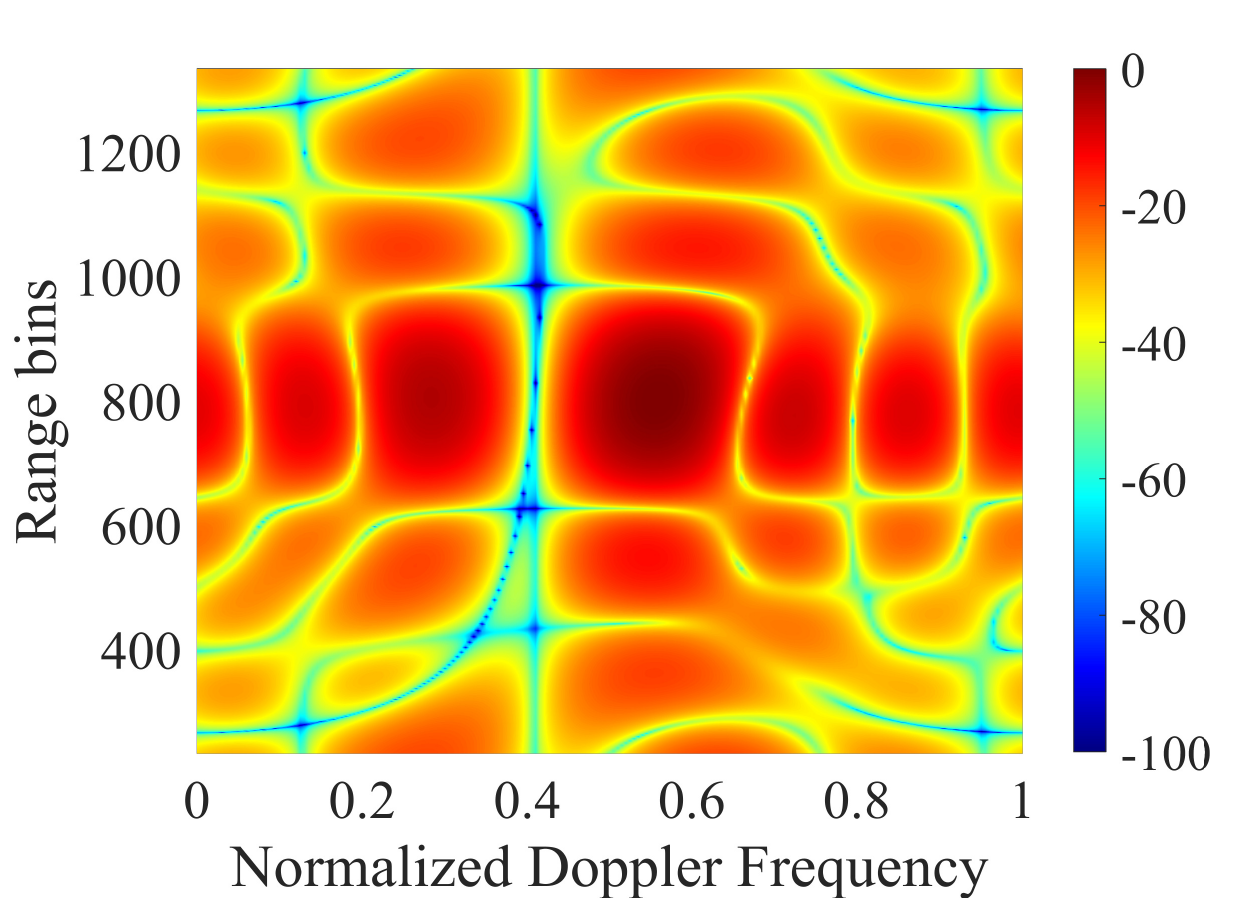}%
\label{fig_10d}}
\hfil
\caption{The range-Doppler two-dimensional spectrum after STAP for FDA-MIMO radar and C-FDA radar with different frequency offset. (a) FDA-MIMO radar. (b) C-FDA ($200$ kHz). (c) C-FDA ($100$ kHz). (d) C-FDA ($50$ kHz).}
\label{fig_10}
\end{figure*}

In Fig.\ref{fig_9}, we show the range-Doppler two-dimensional clutter spectrum for FDA-MIMO radar and C-FDA radar with different frequency offsets. For FDA-MIMO radar, Fig.\ref{fig_9a} illustrates that the range-ambiguous clutter is mixed with secondary range-ambiguous clutter, resulting in a large amount of clutter energy near the range bin indexed at 800. For C-FDA radar, Fig.\ref{fig_9b}, Fig.\ref{fig_9c}, and Fig.\ref{fig_9d} present the clutter spectrum when the frequency offset is $200$ kHz, $100$ kHz, and $50$ kHz. With the decrease of frequency offset, the range ambiguity level of clutter is degraded, resulting in lower clutter energy near the range bin indexed at 800, especially improving the range-dependency of clutter Doppler.

\begin{figure}[!t]
\centering
\includegraphics[width=3.7in]{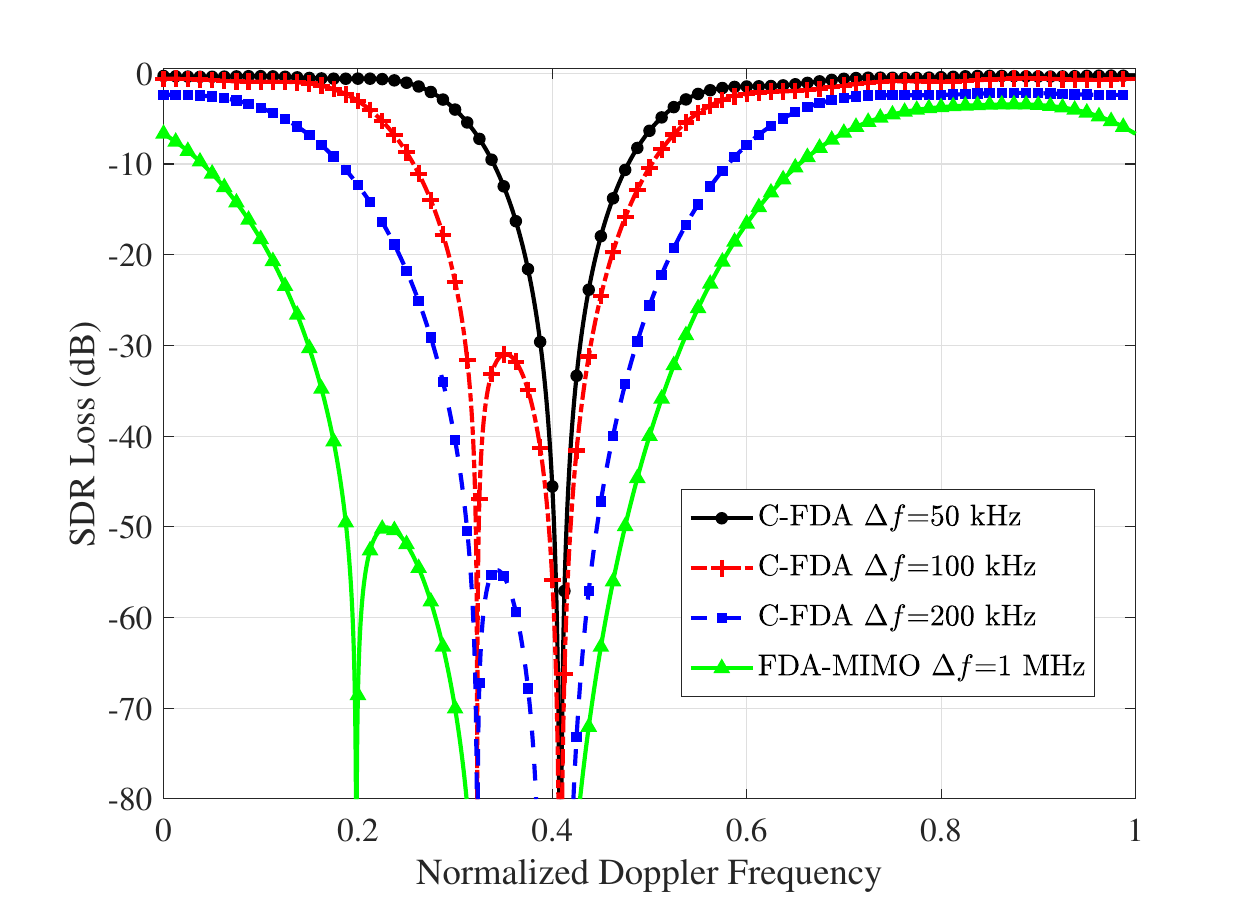}
\caption{Comparison of SDR loss for FDA-MIMO and C-FDA with different frequency offset.}
\label{fig_11}
\end{figure}

In Fig.\ref{fig_10}, we present the range-Doppler two-dimensional spectrum after STAP for FDA-MIMO radar and C-FDA radar with different frequency offsets. Fig.\ref{fig_10a} shows that the clutter notch is significantly wide which can worsen the detection performance of different range bins. With the decrease of frequency offset, Fig.\ref{fig_10b}, Fig.\ref{fig_10c}, and Fig.\ref{fig_10d} demonstrate that the proposed C-FDA radar can effectively resolve the secondary range-ambiguous problem and suppress range-ambiguous clutter while remaining the target energy.

Fig.\ref{fig_11} plots the SDR loss corresponding to the target detection for FDA-MIMO radar and C-FDA radar with different frequency offsets, respectively. It can be indicated that the proposed C-FDA radar outperforms the FDA-MIMO radar on range-ambiguous clutter suppression. In addition to the primary notch at the normalized Doppler frequency of 0.4, FDA-MIMO radar has multiple secondary notches caused by the mixing of range-ambiguous clutter and secondary range-ambiguous clutter, which validates the analysis of \eqref{eq.44}. With the frequency offset decreasing, C-FDA radar has a substantially improved performance on range-ambiguous clutter suppression, which is consistent with Fig.\ref{fig_10}.

\begin{figure}[!t]
\centering
\includegraphics[width=3.7in]{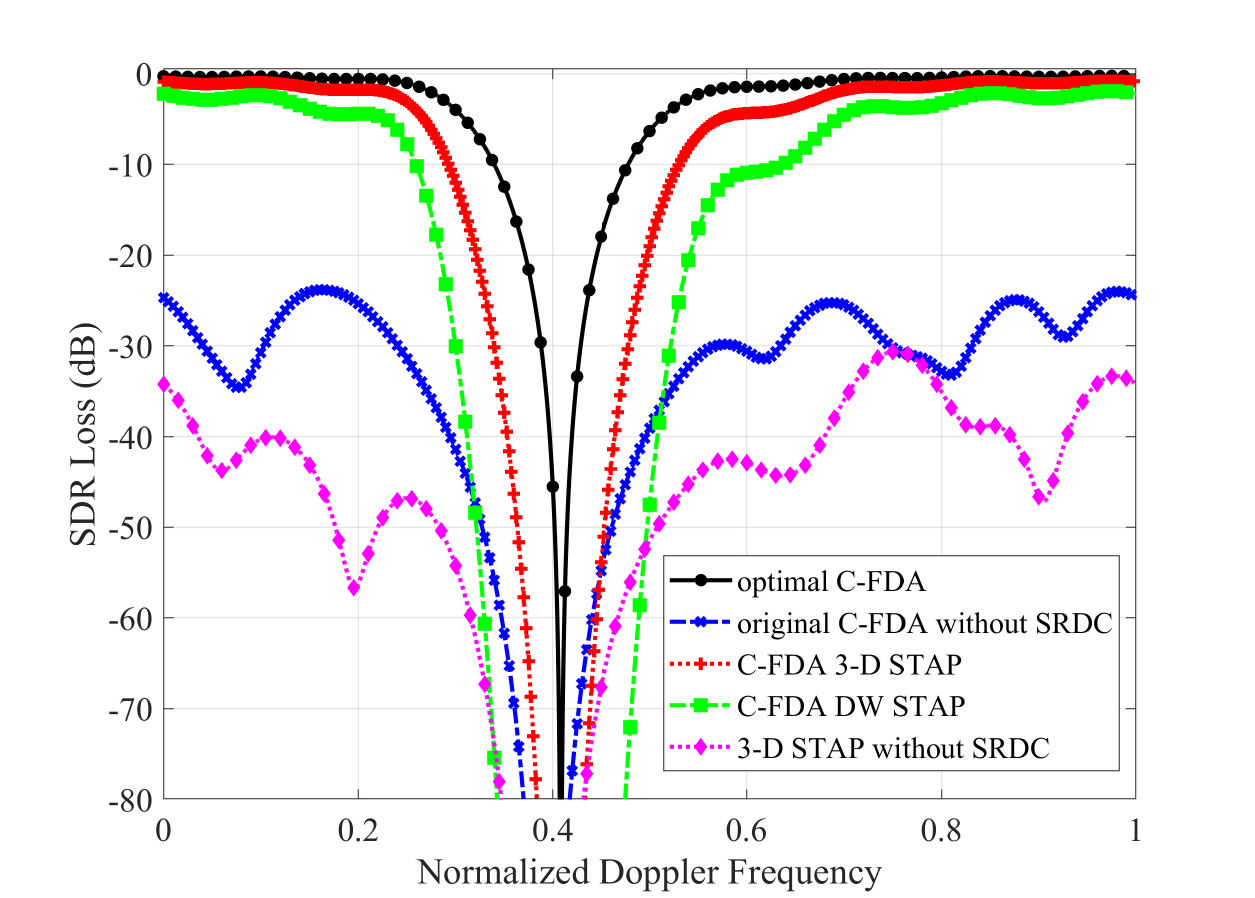}
\caption{Comparison of SDR loss for C-FDA radar with different STAP method.}
\label{fig_12}
\end{figure}

In Fig.\ref{fig_12}, we use typical STAP methods to evaluate the performance of C-FDA radar with a frequency offset of $50$ kHz on range-ambiguous clutter suppression, e.g. azimuth-Doppler-range three-dimensional STAP (3-D STAP) in \cite{ref37} and Doppler warping (DW) STAP \cite{ref38}. In addition, we emphasize the SRDC for C-FDA radar in Fig.\ref{fig_12}. It can be seen that the performance of 3-D STAP is better than DW STAP and the C-FDA-STAP combined with SRDC outperforms C-FDA-STAP without SRDC.

\section{Conclusions}

This paper introduces the concept of C-FDA radar, which combines some benefits of PA and FDA-MIMO radars. The transmitter performs an azimuth beamforming and the receiver performs multi-channel frequency mixing (MFM) and multi-channel matched filtering (MMF). Thanks to this, the C-FDA radar achieves higher coherent processing gain than a PA radar and maintain the range-dependency under the condition of small frequency offset. A generalized space-time-range signal model is derived for the C-FDA radar. By exploiting this model, we show that the C-FDA radar has superior performance than FDA-MIMO radar in terms of interference and clutter cancellation. In particular, we observe that when the frequency offset is much smaller than the signal bandwidth, the C-FDA radar can effectively resolve the secondary range-ambiguous (SRA) problem, whereas the FDA-MIMO radar clutter suppression capabilities are strongly affected by the SRA. The advantages of the C-FDA radar over the FDA-MIMO radar are also demonstrated through the analysis of the output SINR and the SDR loss. Numerical results confirm the goodness of the proposed C-FDA radar in terms of receiver processing output, mainlobe interference suppression, and range-ambiguous clutter suppression.


\setcounter{TempEqCnt}{\value{equation}} 
\setcounter{equation}{52}
\begin{figure*}[ht]
\begin{align}
z_{m,n,k}\left( t \right) =\xi_t \varGamma _{m,n,k}\exp \left\{ j\pi \kappa\left( \tau _{n,k}^{2}-t^2 \right) \right\} \int_{-\infty}^{\infty}{A\left( \frac{x-\tau _{n,k}}{T_p} \right) A\left( \frac{t-x}{T_p} \right) \exp \left\{ j2\pi \kappa x\left( t-\tau _{n,k} \right) \right\}}P_m\left( x \right) \mathrm{d}x
\label{eq.A6}
\end{align}
\hrulefill
\end{figure*}
\setcounter{equation}{\value{TempEqCnt}} 
\setcounter{equation}{38}

\setcounter{TempEqCnt}{\value{equation}} 
\setcounter{equation}{55}
\begin{figure*}[ht]
\begin{subequations}
\begin{align}
z^{(1)}_{m,n,k}\left( t \right) =&\xi_t \varGamma _{m,n,k}\exp \left\{ j\pi \kappa\left( \tau _{n,k}^{2}-t^2 \right) \right\} \int_{-\infty}^{\infty}{A\left( \frac{x-\tau _{n,k}}{T_p} \right) A\left( \frac{t-x}{T_p} \right) \exp \left\{ j2\pi \kappa x\left( t-\tau _{n,k} \right) \right\}}P^{(1)}_m\left( x \right) \mathrm{d}x
\nonumber\\
=&\xi_t \varGamma _{m,n,k}\exp \left\{ j\pi \kappa\left( \tau _{n,k}^{2}-t^2 \right) \right\} \sum_{i=1}^M{\exp \left\{ j2\pi \left( i-m \right) \varDelta f\left( t-\tau _{n,k} \right) \right\}}
\nonumber\\
&\times
\int_{-\infty}^{\infty}{A\left( \frac{x-\tau _{n,k}}{T_p} \right) A\left( \frac{t-x}{T_p} \right) \exp \left\{ j2\pi \kappa x\left( t-\tau _{n,k} \right) \right\}}\mathrm{d}x
\nonumber\\
\approx &\xi_t \varGamma _{m,n,k}\sum_{i=1}^M{\exp \left\{ j2\pi \left( i-m \right) \varDelta f\left( t-\tau _{n,k} \right) \right\} \times E\mathrm{sinc}\left[ \pi \kappa T_p\left( t-\tau _{n,k} \right) \right]}\label{eq.A9a}
\\
z_{m,n,k}^{\left( 2 \right)}\left( t \right) =&\xi_t \varGamma _{m,n,k}\exp \left\{ j\pi \kappa\left( \tau _{n,k}^{2}-t^2 \right) \right\} \int_{-\infty}^{\infty}{A\left( \frac{x-\tau _{n,k}}{T_p} \right) A\left( \frac{t-x}{T_p} \right) \exp \left\{ j2\pi \kappa x\left( t-\tau _{n,k} \right) \right\}}P_{m}^{\left( 2 \right)}\left( x \right) \mathrm{d}x
\nonumber\\
=&\xi_t \varGamma _{m,n,k}\exp \left\{ j\pi \kappa\left( \tau _{n,k}^{2}-t^2 \right) \right\} 
\nonumber\\
\times \int_{-\infty}^{\infty}A&\left( \frac{x-\tau _{n,k}}{T_p} \right) A\left( \frac{t-x}{T_p} \right) \sum_{i=1\&i\ne j}^M{\sum_{j=1}^M{\exp \left\{ j2\pi \left\{ \kappa \left( t-\tau _{n,k} \right) x+\varDelta f\left[ \left( i-j \right) x-\left( i-m \right) \tau _{n,k}-\left( m-j \right) t \right] \right\} \right\}}}\mathrm{d}x
\nonumber\\
\approx &\xi_t \varGamma _{m,n,k}\sum_{i=1\&i\ne j}^M{\sum_{j=1}^M{E}}\mathrm{sinc}\left\{ \pi T_p\left[ \kappa\left( t-\tau _{n,k} \right) +\left( i-j \right) \varDelta f \right] -\pi \varDelta f\left[ \left( i-m \right) \tau _{n,k}-\left( m-j \right) t \right] \right\} 
\label{eq.A9b}
\end{align}
\end{subequations}
\hrulefill
\end{figure*}
\setcounter{equation}{\value{TempEqCnt}} 
\setcounter{equation}{47}




{\appendices
\section{Proof of Theorem 1}
This appendix shows the proof of Theorem \ref{Thm.1}.

Considering that the baseband signal is typical linear frequency modulation (LFM) signal, then
\begin{equation}
\phi\left( t \right) =A\left( \frac{t}{T_p} \right) \exp \left\{ j\pi \kappa t^2 \right\} 
\label{eq.A1}
\end{equation}
where $A\left( {t}/{T_p} \right)$ is a rectangular envelope signal with the time width $T_p$, the bandwidth of $B$, and the unit energy $\int_{T_p}{\left| A\left( {t}/{T_p} \right) \right|}^2\mathrm{d}t=1$. $\kappa=B/T_p$ denotes the frequency modulation ratio. According to the proposed receiver structure mentioned in Section III, the k-th pulse signal received by the n-th element can be expressed as
\begin{align}
y_{n,k}\left( t \right) =&\xi_t A\left( \frac{t-\tau _{n,k}}{T_p} \right) \exp \left\{ j\pi \kappa\left( t-\tau _{n,k} \right) ^2 \right\} 
\label{eq.A2}\\
&\times\sum_{i=1}^M{\exp \left\{ j2\pi \left[ f_c+\left( i-1 \right) \varDelta f \right] \left( t-\tau _{n,k} \right) \right\}}
\nonumber
\end{align}
where $\tau _{n,k}$ is mentioned after \eqref{eq.17}. After the aforementioned MFM, the k-th pulse signal in the m-th channel of the n-th receive element can be written as 
\begin{align}
\tilde{y}_{m,n,k}\left( t \right) =&\xi_t A\left( \frac{t-\tau _{n,k}}{T_p} \right) \exp \left\{ j\pi \kappa\left( t-\tau _{n,k} \right) ^2 \right\} 
\nonumber\\
&\times \exp \left\{ -j2\pi \left[ f_c+\left( m-1 \right) \varDelta f \right] \tau _{n,k} \right\} 
\nonumber\\
&\times \sum_{i=1}^M{\exp \left\{ j2\pi \left( i-m \right) \varDelta f\left( t-\tau _{n,k} \right) \right\}}
\label{eq.A3}
\end{align}
According to \eqref{eq.18} of the MMF, the matched filter function of the m-th channel can be expressed as
\begin{align}
h_m\left( t \right) =A\left( \frac{t}{T_p} \right) \exp \left\{ -j\pi \kappa t^2 \right\} \sum_{i=1}^M{\exp \left\{j2\pi \left( i-m \right) \varDelta ft \right\}}
\label{eq.A4}
\end{align}

The output of MFM-then-MMF in the m-th channel can be calculated by convolution with \eqref{eq.A3} and \eqref{eq.A4}.
\begin{align}
z_{m,n,k}\left( t \right) =&\tilde{y}_{m,n,k}\left( t \right) \circledast h_m\left( t \right) 
\nonumber\\
=&\int_{-\infty}^{\infty}{\tilde{y}_{m,n,k}\left( x \right) \times}h_m\left( t-x \right) \mathrm{d}x
\label{eq.A5}
\end{align}
Substituting \eqref{eq.A3} and \eqref{eq.A4} into \eqref{eq.A5}, we get \eqref{eq.A6}, where
\setcounter{TempEqCnt}{\value{equation}} 
\setcounter{equation}{53}
\begin{subequations}
\begin{align}
\varGamma _{m,n,k}=&\exp \left\{ -j2\pi \left[ f_c+\left( m-1 \right) \varDelta f \right] \tau _{n,k} \right\} \label{eq.A7a}
\\
P_m\left( x \right) =&\sum_{i=1}^M{e^{ j2\pi \left( i-m \right) \varDelta f\left( x-\tau _{n,k} \right)}}\sum_{j=1}^M{e^{ j2\pi \left( j-m \right) \varDelta f\left( t-x \right)}}
\label{eq.A7b}
\end{align}
\end{subequations}
Focusing on $P_m\left( x \right)$, we divide into $P^{(1)}_m(x)$ when $i=j$ and $P^{(2)}_m(x)$ when $i\ne j$, corresponding to $z^{(1)}_{m,n,k}\left( t \right)$ and $z^{(2)}_{m,n,k}\left( t \right)$, respectively.
\begin{subequations}
\begin{align}
&P_{m}^{\left( 1 \right)}\left( x \right) =\sum_{i=1}^M{e^{ j2\pi \left( i-m \right) \varDelta f\left( t-\tau _{n,k} \right)}}
\label{eq.A8a}
\\
&P_{m}^{\left( 2 \right)}\left( x \right) =\sum_{i=1,i\ne j}^M{\sum_{j=1}^M{e^{j2\pi \varDelta f\left[ \left( i-j \right) x-\left( i-m \right) \tau _{n,k}-\left( m-j \right) t \right]}}}\label{eq.A8b}
\end{align}
\end{subequations}
Note that the baseband waveform signal $A(t/T_p)$ within the integral restricts the interval of $t$ to two cases, $\tau _{n,k}\leqslant t\leqslant \tau _{n,k}+T_p$ or $\tau _{n,k}-T_p\leqslant t\leqslant \tau _{n,k}$. Integrating the integration results for these two cases, we explicitly write $z^{(1)}_{m,n,k}\left( t \right)$ and $z^{(2)}_{m,n,k}\left( t \right)$ in \eqref{eq.A9a} and \eqref{eq.A9b}, where $E$ denotes the amplitude coefficient that is discussed in Appendix B.

Precisely, $z^{(1)}_{m,n,k}\left( t \right)$ is summed by $M$ polynomials, where each term can get a peak at $t_{\text{peak}}=\tau_{n,k}$ according to the properties of the sinc function\cite{ref4}. And $z^{(2)}_{m,n,k}\left( t \right)$ is summed by $M\times(M-1)$ polynomials. Thereby, the peak of the ($i$,$j$)-th term can be expressed as 
\setcounter{TempEqCnt}{\value{equation}} 
\setcounter{equation}{56}
\begin{align}
\tau _{\mathrm{peak}}=&\frac{\left( i-m \right) \varDelta f\tau _{n,k}-\left( i-j \right) \varDelta fT_p+\kappa \tau _{n,k}T_p}{\kappa T_p-\left( m-j \right) \varDelta f}
\nonumber\\
=&\frac{\left[ \left( i-m \right) \tau _{n,k}-\left( i-j \right) T_p \right] \varrho +\tau _{n,k}}{1-\left( m-j \right) \varrho}
\label{eq.A10}
\end{align}
where $\varrho =\varDelta f/B$. When $\varDelta f\ll B$, the peaks of each polynomials in $z^{(2)}_{m,n,k}\left( t \right)$ is also at $t_{\text{peak}}=\tau_{n,k}$. Therefore, $z_{m,n,k}\left( t \right)$ can be sampled at $t=\tau_{n,k}$ to get the peak value, containing information related to the echo signal. Integrating \eqref{eq.A7a}, \eqref{eq.A9a}, and \eqref{eq.A9b}, 
\begin{align}
z_{m,n,k}= E\xi_t \exp \left\{ -j2\pi \left[ f_c+\left( m-1 \right) \varDelta f \right] \tau _{n,k} \right\} 
\label{eq.A11}
\end{align}
Then substituting \eqref{eq.13a}, \eqref{eq.13c}, and \eqref{eq.13d} associated with $\tau _{n,k}$ and ignoring the quadratic term,
\begin{align}
z_{m,n,k}\approx & E\xi_t e^{ j2\pi \left( m-1 \right) f_R }e^{ j2\pi \left( n-1 \right) f_{\varphi}} e^{ j2\pi \left( k-1 \right) f_D } 
\label{eq.A12}
\end{align}
Thereby, for the proposed C-FDA radar, the $MNK$-dimensional target data vector $\boldsymbol{t}_{\mathrm{C-F}}$ with $z_{m,n,k}$ as the $(m,n,k)$-th element can be expressed as 
\begin{align}
\boldsymbol{t}_{\mathrm{C}-\mathrm{F}}=&\left[ \begin{matrix}
        z_{1,1,1}&        \cdots&      z_{m,n,k}&       \cdots&     z_{M,N,K}\\
\end{matrix} \right] ^{{T}}
\nonumber\\
=&E\xi _t\boldsymbol{a}_R\left( R_t \right) \otimes \boldsymbol{a}_{\varphi}\left( \varphi _t,\theta _t \right) \otimes \boldsymbol{a}_D\left( f_D \right) 
\label{eq.A13}
\end{align}
where $\boldsymbol{a}_R\left( R_t \right)$, $\boldsymbol{a}_{\varphi}\left( \varphi _t,\theta _t \right)$, and $\boldsymbol{a}_D\left( f_D \right)$ are the transmit range-dependent frequency, the receive spatial frequency, and the Doppler frequency, respectively, easily derived by $f_R$, $f_\varphi$, and $f_D$ as shown in \eqref{eq.20a}, \eqref{eq.20b}, and \eqref{eq.20c}.

\section{Proof of Theorem 2}
This appendix shows the proof of Theorem \ref{Thm.2}, especially the explaination of $E$ in \eqref{eq.19}.

Assume the Fourier transform of $\phi \left( t \right)$ is $Z(f)$. Accooding to Parseval - Gutzmer theorem\cite{ref1}\cite{ref4}, 
\begin{equation}
\int_{-\infty}^{+\infty}{\left| Z\left( f \right) \right|^2\mathrm{d}f}=\int_{T_p}{\left| \phi \left( t \right) \right|^2\mathrm{d}t}=1
\label{eq.B1}
\end{equation}
The Fourier transform of the signal transmitted by the m-th element in \eqref{eq.2} can be expressed as
\begin{equation}
\mathcal{F} \left\{ u_m\left( t \right) \right\} =Z\left[ f-\left( m-1 \right) \varDelta f-f_c \right] 
\label{eq.B2}
\end{equation}
Then the Fourier transform of the received signal related to the n-th receive element and k-th pulse can be expressed as 
\begin{align}
\mathcal{F} \left\{ y_{n,k}\left( t \right) \right\} =&\xi_t\exp \left\{ -j2\pi f\tau _{n,k} \right\}
\nonumber\\ 
&\times \sum_{i=1}^M{Z\left[ f-\left( i-1 \right) \varDelta f -f_c\right]}
\label{eq.B3}
\end{align}
After the MFM in the m-th channel, then
\begin{align}
\mathcal{F} \left\{ \tilde{y}_{m,n,k}\left( t \right) \right\} =&\xi\exp \left\{ -j2\pi \left[ f_c+\left( m-1 \right) \varDelta f-f \right] \tau _{n,k} \right\} 
\nonumber\\
&\times\sum_{i=1}^M{Z\left[ f-\left( m-i \right) \varDelta f \right]}
\label{eq.B4}
\end{align}
Note that the Fourier transform of the matched filtering function in \eqref{eq.A4} is
\begin{equation}
\mathcal{F} \left\{ h_m\left( t \right) \right\} =\sum_{j=1}^M{Z^*\left[f-\left( m-j \right) \varDelta f \right]}
\label{eq.B5}
\end{equation}
The output of the matched filtering in the m-th channel can be expressed as 
\begin{equation}
Q_m\left( f \right) =\mathcal{F} \left\{ \tilde{y}_{m,n,k}\left( t \right) \right\} \times \mathcal{F} \left\{ h_m\left( t \right) \right\} =G_m\left( f \right) H_m\left( f \right) 
\label{eq.B6}
\end{equation}
where 
\begin{subequations}
\begin{align}
&G_m\left( f \right)=\exp \left\{ j2\pi \left[ f-f_c+\left( m-1 \right) \varDelta f \right] \tau _{n,k} \right\} 
\label{eq.B7a}
\\
&H_m\left( f \right)=\sum_{i=1}^M{Z\left[ f-\left( m-i \right) \varDelta f \right]}\sum_{j=1}^M{Z^*\left[ f-\left( m-j \right) \varDelta f \right]}\label{eq.B7b}
\end{align}
\end{subequations}
We focus on $H_m(f)$ that strongly influences the amplitude of the m-th channel output. 
\begin{equation}
\int_{-\infty}^{+\infty}{\left| Q_m\left( f \right) \right|^2\mathrm{d}f}=\int_{-\infty}^{+\infty}{\left| H_m\left( f \right) \right|^2\mathrm{d}f}
\label{eq.B8}
\end{equation}
Define an auxiliary vector $\boldsymbol{b}$ and $\boldsymbol{B}=\boldsymbol{b}\boldsymbol{b}^T$,
\begin{equation}
\boldsymbol{b}=\left[ \begin{matrix}
        Z\left[ f-\left( m-1 \right) \varDelta f \right]&               \cdots&            Z\left[ f-\left( m-M \right) \varDelta f \right]\\
\end{matrix} \right] ^T
\label{eq.B9}
\end{equation}
then $H_m\left( f \right)$ is the sum of all entries in $\boldsymbol{B}$. Apparently, the diagonal entries in $\boldsymbol{B}$ is the maximum of each row, thereby
\begin{equation}
\mathrm{tr}\left( \boldsymbol{B} \right) <H_m\left( f \right) \leqslant M\mathrm{tr}\left( \boldsymbol{B} \right) 
\label{eq.B10}
\end{equation}
where the condition for the equal to hold is $\varDelta f=0$. Therefore, the m-th channel output energy is 
\begin{equation}
M<\int_{-\infty}^{+\infty}{\left| Q_m\left( f \right) \right|^2\mathrm{d}f}\leqslant M^2
\label{eq.B11}
\end{equation}
and the amplitude of the m-th channel output is $\sqrt{M}<E\leqslant M$. 

Similar to \eqref{eq.11a} and \eqref{eq.11b}, the coherent array gain of C-FDA radar can be calculated as
\begin{equation}
\varOmega _{\mathrm{C}-\mathrm{F}}=\frac{\sigma _{\mathrm{t}}^{2}\left| \boldsymbol{w}_{\mathrm{C}-\mathrm{F}}^{H}\boldsymbol{t}_{\mathrm{C}-\mathrm{F}} \right|^2}{\boldsymbol{w}_{\mathrm{C}-\mathrm{F}}^{H}\boldsymbol{R}_{\mathrm{n}}^{\left( \mathrm{C}-\mathrm{F} \right)}\boldsymbol{w}_{\mathrm{C}-\mathrm{F}}}=E^4NK\cdot \mathrm{SNR}_{\mathrm{in}}
\label{eq.B12}
\end{equation}
then the comparison between PA radar, FDA-MIMO radar, and C-FDA radar is 
\begin{equation}
M\varOmega _{\mathrm{F}-\mathrm{M}}=\varOmega _{\mathrm{PA}}=M^2NK\cdot \mathrm{SNR}_{\mathrm{in}}<\varOmega _{\mathrm{C}-\mathrm{F}}\leqslant M^2\varOmega _{\mathrm{PA}}
\label{eq.B13}
\end{equation}
}

%

\newpage

\vfill

\end{document}